\documentclass[12pt]{article}
\makeatletter
\providecommand{\@LN}[2]{}
\makeatother
\usepackage{times}
\usepackage{bm,amsmath,amssymb,amsthm}
\usepackage{natbib,enumitem,graphics,graphicx,url}

\newcommand{\pkg}[1]{$\mathsf{ #1}$}

\newcommand{\bB}{{\bf B}}
\newcommand{\bx}{{\bf x}}

\newcommand{\bX}{{\bf X}}
\newcommand{\br}{{\bf r}}
\newcommand{\bff}{{\bf f}}
\newcommand{\bG}{\boldsymbol{G}}
\newcommand{\bY}{{\bf Y}}

\newcommand{\by}{{\bf y}}
\newcommand{\bz}{{\bf z}}
\newcommand{\bb}{{\bf b}}
\newcommand{\bF}{{\bf F}}
\newcommand{\bQ}{{\bf Q}}

\newcommand{\bD}{{\bf D}}

\newcommand{\bI}{{\bf I}}

\newcommand{\bM}{{\bf M}}

\newcommand{\bw}{{\bf w}}
\newcommand{\bg}{{\bf g}}
\newcommand{\bW}{{\bf W}}

\newcommand{\bu}{{\bf u}}
\newcommand{\bE}{{\bf E}}
\newcommand{\bzero}{{\bf 0}}

\DeclareMathOperator{\Tr}{Tr}

\def\independenT#1#2{\mathrel{\rlap{$#1#2$}\mkern2mu{#1#2}}}

\newcommand{\bbeta}{\boldsymbol{\beta}}
\newcommand{\bmu}{\boldsymbol{\mu}}

\newcommand{\bSigma}{\boldsymbol{\Sigma}}

\newcommand{\btau}{\boldsymbol{\tau}}
\newcommand{\brho}{\boldsymbol{\rho}}

\newcommand{\hB}{\hat{B}}

\newcommand{\hbB}{\widehat{\bB}}

\newcommand{\beps}{\boldsymbol{\epsilon}}
\newcommand{\eps}{\epsilon}
\newcommand{\hbeps}{\hat{\boldsymbol{\epsilon}}}
\newcommand{\bOmega}{\boldsymbol{\Omega}}

\newcommand{\real}{I\kern-0.37emR}

\newtheorem{example}{Example}[section]

\newtheorem{theorem}{Theorem}

\newtheorem{definition}{Definition}
\newtheorem{lemma}{Lemma}
\newtheorem{corollary}{Corollary}
\newtheorem{remark}{Remark}
\newtheorem{proposition}{Proposition}
\newtheorem{condition}{Condition}

\newcommand{\convinprob}{\stackrel{P}{\rightarrow}}
\newcommand{\convindist}{\stackrel{\mathcal{D}}{\rightarrow}}
\newcommand\independent{\protect\mathpalette{\protect\independenT}{\perp}}
\def\independenT#1#2{\mathrel{\rlap{$#1#2$}\mkern2mu{#1#2}}}

\newcommand{\bc}{{\bf c}}

\title{\bf A Projection-based Conditional Dependence Measure with Applications to High-dimensional Undirected Graphical Models \thanks{Corresponding Authors: Yang Feng, Department of Statistics, Columbia University, New York, NY 10027, E-mail: yang.feng@columbia.edu, Phone: 212-851-2139; Lucy Xia, Department of Statistics, Stanford University, Stanford, CA 94305, E-mail: lucyxia@stanford.edu, Phone: 650-497-8146.}}
\date{}
\author{Jianqing Fan$^1$, Yang Feng$^2$ and Lucy Xia$^3$\\
{\it $^1$Princeton University, $^2$Columbia University and $^3$Stanford University}}
\begin{document}
\maketitle

\begin{abstract}
Measuring conditional dependence is an important topic in econometrics with broad applications including graphical models. Under a factor model setting, a new conditional dependence measure based on projection is proposed. The corresponding conditional independence test is developed with the asymptotic null distribution unveiled where the number of factors could be high-dimensional. It is also shown that the new  test has control over the asymptotic type I error and can be calculated efficiently. A generic method for building dependency graphs without Gaussian assumption using the new test is elaborated. Numerical results and real data analysis show the superiority of the new method.\end{abstract}

\section{Introduction}

Rapid development in technology continuously floods us with high-dimensional data nowadays. One interesting question we want to answer is their conditional dependency structure. Mathematically, let $\bz=(z^{(1)},\dots,z^{(d)})$ be a $d$-dimensional random vector, which represents our data. We would like to know whether $z^{(i)}\independent z^{(j)}|\bz\setminus \{z^{(i)},z^{(j)}\}$ for $i\neq j\in\{1,\dots, d\}$ or if there are other factors $\bff$ associated with $\bz$, then whether $z^{(i)}\independent z^{(j)}|\bff$. For visualization purpose, sometimes people use a graph to represent such a structure; in other words, suppose each coordinate in the multivariate data is a vertex, then an edge will be drawn if conditioning on some factors we choose, the two vertices are dependent. Producing such a graph can help us understand data and has been an important topic in fields including finance, signal processing, bioinformatics and network modeling \citep{wainwright2008graphical}.

Intrinsically, if we look at the nodes pair by pair, this is a testing problem. In general, we can write our goal as testing whether $\bx$ and $\by$ are independent given $\bff$, i.e.,
\begin{align}\label{eq::indep-test-ori}
  H_0:  \bx\independent\by|\bff ,%\mbox{ vs. } H_1:  \text{not $H_0$}.
\end{align}
where $\bx$, $\by$ and $\bff$ are random vectors with possibly different dimensions. For such conditional independence tests, there has been abundant literature, especially in economics.  \cite{linton1997conditional} proposed two nonparametric tests based on a generalization of the empirical distribution function; however, a complicated bootstrap procedure is needed to calculate critical values of the test, which leads to limited practical value. \cite{su2007consistent, su2008nonparametric, su2014testing} and \cite{wang2015conditional} proposed conditional independence tests based on Hellinger distance, empirical likelihood, conditional moments and conditional characteristic function, respectively. However,  as many of the recently  available datasets are of high-dimension, the computation for these tests becomes prohibitive and strongly limits the practical value. Another related work is \cite{sen2014testing}, where the focus is on testing the independence between the error and the predictor variables in the linear regression problem. 

Our starting point is a relatively ``general" model on $\{\bx,\by,\bff\}$. In particular, suppose $\{(\bx_i, \by_i, \bff_i, \beps_{i,x}, \beps_{i,y}), i=1,\dots, n\}$ are i.i.d. realizations of $(\bx, \by, \bff, \beps_x, \beps_y)$, which are generated from the following model:
\begin{align}\label{eq::model-general}
\bx = \bG_x(\bff) +\beps_{x}, \quad \by=\bG_y(\bff)+\beps_{y},
\end{align}
where $\bff$ is the $K$-dimensional common factors, $\bG_x$ and $\bG_y$ are general mappings from $\mathbb{R}^K$ to $\mathbb{R}^p$ and $\mathbb{R}^q$, respectively. The observed data are $\{(\bx_i, \by_i, \bff_i), i=1,\ldots,n\}$. Here, for simplicity and tractability, we assume independence between $(\beps_{x}, \beps_{y})$ and $\bff$. Such kind of models shed light for another route to solve the issues, nonparametric regression. The idea is intuitive: to test \eqref{eq::indep-test-ori} is the same as testing $\beps_{x}\independent \beps_{y}$ under \eqref{eq::model-general}, which naturally leads to  a two-step procedure.  Since  $\{(\beps_{i,x}, \beps_{i,y}), i=1,\dots, n\}$ are not observed, in Step 1, we estimate the residuals. In this regard, we assume the dimensions $p$ and $q$ to be fixed while the number of factors $K$ could diverge to infinity. Step 2, we apply an independence test on the estimated residuals. These two steps constitute our new conditional independence test and we will unveil the asymptotic properties for this new test statistic. Let's briefly preview the procedure in the following two paragraphs.

In Step 1, ideally, a fully nonparametric projection on $\bff$ (e.g., local polynomial regression) would consistently recover the random errors asymptotically under certain smoothness assumptions on $\bG_x$ and $\bG_y$, when $K$ is fixed. However, it becomes challenging when $K$ diverges due to the curse of dimensionality if no structural assumptions are made on $\bG_x$ and $\bG_y$.   As a result, in this paper, we will study  two  cases where $\bG_x$ and $\bG_y$ are linear functions (factor models) in Section \ref{subsec::T} and where $\bG_x$ and $\bG_y$ are additive functions in Section \ref{subsec:: Functional Projection} when $K$ diverges. Further relaxed models might be available for future work, but we don't focus on them in this paper.

To complete our proposal, after estimating the residuals in Step 1, we still need to find a suitable measure of dependence between random variables/vectors in Step 2. In this regard, many different measures of dependence have  been proposed. Some of them rely heavily on Gaussian assumptions, such as Pearson correlation, which measures linear dependence and the uncorrelatedness is equivalent to independence only when the joint distribution is Gaussian; or Wilks Lambda \citep{wilks1935independence}, where normality is adopted to calculate the likelihood ratio. To deal with non-linear dependence and non-Gaussian distribution, statisticians have proposed rank-based correlation measures, including Spearman's $\rho$  and Kendall's $\tau$, which are more robust than Pearson correlation against deviations from normality. However, these correlation measures are usually only effective for monotonic types of dependence. In addition, under the null hypothesis that two variables are independent, no general statistical distribution of the coefficients associated with these measures has been derived. Other related works include \cite{hoeffding1948non},  \cite{blomqvist1950measure}, \cite{blum1961distribution}, and some methods described in \cite{hollander2013nonparametric} and \cite{anderson1958introduction}. Taking these into consideration, distance covariance \citep{szekely2007measuring} was introduced to address these deficiencies. The major benefits of distance covariance are: first, zero distance covariance implies independence, and hence it is a true dependence measure. Second, distance covariance can measure the dependence between any two vectors which potentially are of different dimensions. Recently, \cite{huo2016fast} proposed a fast computation method for distance covariance. Due to these advantages, we will focus on  distance covariance in this paper as our measure of dependence.

So far, we complete a rough description of the newly proposed conditional dependence measure; and we are able to build conditional dependency graphs by conducting this test edge by edge. We would like to make two remarks here to help readers connect the dots between our work and some other existing related topics/works.

First, let us look at the connection to undirected graphical models.  Undirected graphical models (UGM) has been a popular topic in econometrics in the past decade. It studies the {\it ``internal"} conditional dependency structure of a multivariate random vector.  To be more explicit, again let $\bz=(z^{(1)},\dots,z^{(d)})$ be the $d$-dimensional random vector of interest. We denote the undirected graph corresponding to $\bz$ by $(V,E)$, where vertices $V$ correspond to components of $\bz$ and edges $E=\{e_{ij},1\leq i\neq j\leq d\}$ indicate whether  node $z^{(i)}$ and $z^{(j)}$ are conditionally independent given the remaining nodes. In particular, the edge $e_{ij}$ is absent if and only if $z^{(i)}\independent z^{(j)}|\bz\setminus \{z^{(i)},z^{(j)}\}$. Therefore, UGM is a nature application of our measure if we take $\bff=\bz\setminus \{z^{(i)},z^{(j)}\}$ in our test. One intensively studied sub-field is GGM (Gaussian graphical model) where $\bz$ is assumed to follow a multivariate Gaussian distribution with mean $\bmu$ and covariance matrix $\bSigma$. This extra assumption is desirable since then the precision matrix $\bOmega=(w_{ij})_{d\times d}=\bSigma^{-1}$ captures exactly the conditional dependency graph; that is, $w_{ij}= 0$ if and only if $e_{ij}$ is absent \citep{lauritzen1996graphical, edwards2000introduction}. Therefore, under the Gaussian assumption, this problem reduces to the estimation of precision matrix, where a rich literature on model selection and parameter estimation can be found in both low-dimensional and high-dimensional settings, including \cite{dempster1972covariance}, \cite{drton2004model}, \cite{meinshausen2006high}, \cite{friedman2008sparse},  \cite{fan2009network},  \cite{cai2011constrained}, \cite{liu2013gaussian}, \cite{chen2014selection}, \cite{ren2015asymptotic}, \cite{jankova2015confidence} and \cite{yu2017learning}. With simple derivations, it's easy to check that GGM fits into our framework with $\bG$ being linear and $\boldsymbol{\epsilon}$ having Gaussian distributions. Therefore, with linear projection in Step 1 and distance covariance in Step 2, our proposed conditional measure solves GGM. It's worth noting that, indeed, in Step 2, choosing Pearson correlation will solve GGM as well; choosing distance covariance gives more flexibility since we don't assume normality on $\boldsymbol{\epsilon}$ and potentially we can solve non-Gaussian UGMs. Another interesting work is \cite{voorman2013graph}, where a semi-parametric method was introduced for graph estimation.

Second, we examine the link to  factor models. As explained in the last paragraph, UGM is a case with $\bff$ being internal factors, in other words, part of the interested vector $\bz$. Another scenario of our framework is the case when $\bff$ are external, and this is closely related to factor models. As an example, in the Fama-French three-factor model, the return of each stock can be considered as one node in the graph we want to build and $\bff$ are the chosen three-factors. This example will be further elaborated in Section \ref{sec::Realdata}. Therefore, the factors $\bff$ are considered as external since they are not part of the individual stock returns. Another interesting application is discussed in \cite{stock2002macroeconomic}, where external factors are aggregated macroeconomic variables, and the nodes are disaggregated macroeconomic variables.

With the above two remarks, we see our proposed test cover some of the existing topics as by-products. We summarize the main contribution of this paper here. First, under model \eqref{eq::model-general}, we propose a computationally efficient conditional independence test. Both the response vectors and the common factors can be of different dimensions and the number of the factors could grow to infinity with sample size.  Second, we apply this test to build conditional dependency graph (internal factors) and covariates-adjusted dependency graph (external factors).

The rest of this paper is organized as follows. In Section \ref{sec::Methods}, we  present our new procedure for testing conditional independence via projected distance covariance (P-DCov)  and describe how to construct conditional dependency graphs based on the proposed test. Section \ref{sec::TheoreticalP} gives theoretical properties including the asymptotic distribution of the test statistic under the null hypothesis as well as the type I error guarantee. Section \ref{sec::NumericalS} contains extensive numerical studies and Section \ref{sec::Realdata} demonstrates the performance of P-DCov via a financial data set. We conclude the paper with a short discussion in Section \ref{sec::Disc}.  Several technical lemmas and all proofs are relegated to the appendix.

\section{Methods}\label{sec::Methods}
\subsection{A brief review of distance covariance}
First, we introduce some notations. For a random vector $\bz$, $\|\bz\|$ and $\|\bz\|_1$ represent its Euclidean norm and $\ell_1$ norm, respectively. A collection of $n$ i.i.d. observations of $\bz$ is denoted as $\{\bz_1,\dots,\bz_n\}$, where $\bz_k=(z^{(1)}_k,\dots,z^{(d)}_k)^T$ represents the $k$-th observation. For any matrix $\bM$,  $\|\bM\|_F$, $\|\bM\|$ and  $\|\bM\|_{\max}$ denote its Frobenius norm, operator norm and max norm, respectively. $\|\bM\|_{a,b}$ is the $(a,b)$ norm defined as the $\ell_b$ norm of the vector consisting of column-wise $\ell_a$ norm of $\bM$. Furthermore, $a\wedge b$ represents $\min\{a , b\}$ and $a\vee b$ represents $\max\{a,  b\}$.

As an important tool, distance covariance is briefly reviewed in this section with further details available in \cite{szekely2007measuring}. We introduce several definitions as follows.
\begin{definition}($w$-weighted $L_2$ norm)
Let $c_d=\frac{\pi^{(d+1)/2}}{\Gamma((d+1)/2)}$, for any positive integer $d$, where $\Gamma$ is the Gamma function. Then for function $\gamma$ defined on $\mathbb{R}^p\times\mathbb{R}^q$, the $w$-weighted $L_2$ norm of $\gamma$ is defined by
\[
\|\gamma(\btau,\brho)\|^2_w=\int_{\mathbb{R}^{p+q}}|\gamma(\btau,\brho)|^2 w(\btau,\brho) d\btau d\brho,\ \ \text{where}\ \ w(\btau,\brho)=(c_p c_q \|\btau\|^{1+p} \|\brho\|^{1+q})^{-1}.
\]
\end{definition}

\begin{definition} (Distance covariance) The distance covariance between random vectors $\bx\in\mathbb{R}^p$ and $\by\in\mathbb{R}^q$ with finite first moments is the nonnegative number  $\mathcal{V}(\bx,\by)$ defined by
\[
\mathcal{V}^2(\bx,\by)=\|g_{\bx,\by}(\btau,\brho)-g_{\bx}(\btau)g_{\by}(\brho)\|^2_w,
\]
where $g_{\bx}$, $g_{\by}$ and $g_{\bx,\by}$ represent the characteristic functions of $\bx$, $\by$ and the joint characteristic function of $\bx$ and $\by$, respectively.
\end{definition}

Suppose we observe random sample $\{(\bx_k,\by_k):k=1,\dots,n\}$ from the joint distribution of $(\bx,\by)$. We denote $\bX = (\bx_1,\bx_2,\dots,\bx_n)$ and $\bY = (\by_1,\by_2,\dots, \by_n)$.

\begin{definition} (Empirical distance covariance) The empirical distance covariance between  samples $\bX$ and $\bY$
  is the nonnegative random variable $\mathcal{V}_n(\bX,\bY)$  defined by
\[
\mathcal{V}^2_n(\bX,\bY)=S_1(\bX,\bY)+S_2(\bX,\bY)-2S_3(\bX,\bY),
\]
where
\begin{align*}
S_1(\bX,\bY)&=\frac{1}{n^2}\sum_{k,l=1}^n \|\bx_k-\bx_l\|\|\by_k-\by_l\|,\ S_2(\bX,\bY)=\frac{1}{n^2}\sum_{k,l=1}^n \|\bx_k-\bx_l\|\frac{1}{n^2}\sum_{k,l=1}^n\|\by_k-\by_l\|,
\cr
S_3(\bX,\bY)&=\frac{1}{n^3}\sum_{k=1}^n\sum_{l,m=1}^n \|\bx_k-\bx_l\|\|\by_k-\by_m\|.
\end{align*}

%\begin{align*}
%S_1&=\frac{1}{n^2}\sum_{k,l=1}^n |\bx_k-\bx_l||\by_k-\by_l|,\cr
%S_2&=\frac{1}{n^2}\sum_{k,l=1}^n |\bx_k-\bx_l|\frac{1}{n^2}\sum_{k.l=1}^n|\by_k-\by_l|,\cr
%S_3&=\frac{1}{n^3}\sum_{k=1}^n\sum_{l,m=1}^n |\bx_k-\bx_l||\by_k-\by_m|.
%\end{align*}
\end{definition}

With above definitions, Lemma \ref{thm::old-consistency} depicts the consistency of $\mathcal{V}_n(\bX,\bY)$ as an estimator of $\mathcal{V}(\bx, \by)$. Lemma \ref{thm::old-asymptotic} shows the asymptotic distribution of $\mathcal{V}_n(\bX,\bY)$ under the null hypothesis that $\bx$ and $\by$ are independent. Corollary \ref{coro::old-teststat} reveals properties of the test statistic $n\mathcal{V}^2_n/S_2$ proposed in \cite{szekely2007measuring}.
\begin{lemma}\label{thm::old-consistency}(Theorem 2 in \cite{szekely2007measuring})
Assume that $\mathbb{E}(\|\bx\|+\|\by\|)<\infty$, then almost surely
\[
\lim_{n\rightarrow \infty} \mathcal{V}_n(\bX,\bY)=\mathcal{V}(\bx,\by).
\]
\end{lemma}

\begin{lemma}\label{thm::old-asymptotic}(Theorem 5 in \cite{szekely2007measuring})
Assume that $\bx$ and $\by$ are independent, and $\mathbb{E}(\|\bx\|+\|\by\|)<\infty$, then as $n\rightarrow \infty$,
\[
n\mathcal{V}^2_n(\bX,\bY)\stackrel{D}{\rightarrow}\|\zeta(\btau,\brho)\|^2_w,
\]
where $\stackrel{D}{\rightarrow}$ represents convergence in distribution and $\zeta(\cdot , \cdot )$ denotes a complex-valued centered Gaussian random process with covariance function
\[
R(\bu,\bu_0)=(g_{x}(\btau-\btau_0)-g_{x}(\btau)\overline{g_{x}(\btau_0)})(g_{y}(\brho-\brho_0)-g_{y}(\brho)\overline{g_{y}(\brho_0)}),
\]
in which $\bu=(\btau,\brho)$, $\bu_0=(\btau_0,\brho_0)$.
\end{lemma}

\begin{corollary}\label{coro::old-teststat}(Corollary 2 in \cite{szekely2007measuring})
Assume that $\mathbb{E}(\|\bx\|+\|\by\|)<\infty$.
\begin{enumerate}
\item If $\bx$ and $\by$ are independent, then as $n\rightarrow \infty$, $n\mathcal{V}^2_n(\bX,\bY)/S_2(\bX,\bY)\stackrel{D}{\rightarrow}Q$ with $Q\overset{D}{=}\sum_{j=1}^{\infty}\lambda_j Z^2_j$, where $Z_j\overset{i.i.d}{\sim}\mathcal{N}(0,1)$ and $\{
    \lambda_j\}$ are non-negative constants depending on the distribution of $(\bx,\by)$; $\mathbb{E}(Q)=1$.
\item If $\bx$ and $\by$ are dependent, then  as $n\rightarrow \infty$, $n\mathcal{V}^2_n(\bX,\bY)/S_2(\bX,\bY)\stackrel{P}{\rightarrow}\infty$.
\end{enumerate}
\end{corollary}

\subsection{Conditional independence test via projected distance covariance (P-DCov)\label{subsec::T}}
Here, we consider the case where $\bG_x$ and $\bG_y$ are linear in \eqref{eq::model-general}, which leads to the  following factor model setup:
\begin{align}\label{eq::model}
\bx = \bB_x \bff +\beps_{x}, \quad \by=\bB_y \bff+\beps_{y},
\end{align}
where $\bB_x$ and $\bB_y$ are factor loading matrices of dimension $p\times K$ and $q\times K$ respectively, and $\bff$ is the $K$-dimensional vector of common factors.  Here, we assume $p$ and $q$ are fixed,  the number of common factors $K$ could grow to infinity and the matrices $\bB_x$ and $\bB_y$ are sparse to reflect that $\bx$ and $\by$ only depend on several important factors.  As a result, we will impose regularization on the estimation of $\bB_x$ and $\bB_y$. Now, we are in the position to propose a test for problem \eqref{eq::indep-test-ori}. We first provide an estimate for the idiosyncratic components $\beps_{x}$ and $\beps_{y}$, and then calculate distance covariance between the estimates. More generally, we project  $\bx$ and $\by$ onto the space orthogonal to the linear space spanned by $\bff$ and evaluate the dependency between the projected vectors.
The conditional independence test is summarized in the following steps.

%\noindent\fbox{\it Procedure of constructing  }

\noindent{\bf Step 1:}
Estimate factor loading matrices $\bB_x$ and $\bB_y$ by the penalized least square (PLS) estimators $\hbB_x$ and $\hbB_y$ defined as follows.
\begin{align}
\hbB_x &= \arg\min_{\bB} \frac{1}{2}\|\bX-\bB\bF\|_F^2+\sum_{j,k}p_{\lambda_1}(|B_{jk}|),\\
\hbB_y &= \arg\min_{\bB}\frac{1}{2} \|\bY-\bB\bF\|_F^2+\sum_{j,k}p_{\lambda_2}(|B_{jk}|),
\end{align}
where $\bX=(\bx_1, \bx_2, \dots, \bx_n)$, $\bY=(\by_1, \by_2, \dots, \by_n)$, $\bF=(\bff_1, \bff_2, \dots, \bff_n)$, $p_{\lambda}(\cdot)$ is the penalty function with penalty level $\lambda$.

\noindent{\bf Step 2:} Estimate the error vectors $\beps_{i,x}$ and $\beps_{i,y}$ by
\begin{align*}
\hbeps_{i,x} &=\bx_i-\hbB_x\bff_i=(\bB_x-\hbB_x)\bff_i+\beps_{i,x},\cr
\hbeps_{i,y} &=\by_i-\hbB_y\bff_i=(\bB_y-\hbB_y)\bff_i+\beps_{i,y},\ \ i=1,\dots,n.
\end{align*}

\noindent{\bf Step 3:} Define the estimated error matrices $\widehat\bE_x = (\hat\beps_{1,x}, \dots, \hat\beps_{n,x})$ and $\widehat\bE_y = (\hat\beps_{1,y}, \dots, \hat\beps_{n,y})$. Calculate the empirical distance covariance between $\widehat\bE_x$ and $\widehat\bE_y$ as
\[
\mathcal{V}^2_n(\widehat\bE_x,\widehat\bE_y)=S_1(\widehat\bE_x,\widehat\bE_y)+S_2(\widehat\bE_x,\widehat\bE_y)-2S_3(\widehat\bE_x,\widehat\bE_y).
\]

\noindent{\bf Step 4:} Define the P-DCov test statistic as $T(\bx,\by,\bff)=n\mathcal{V}^2_n(\widehat\bE_x,\widehat\bE_y)/S_2(\widehat\bE_x,\widehat\bE_y)$.

\noindent{\bf Step 5:} With a predetermined significance level $\alpha$, we reject the null hypothesis when $T(\bx,\by,\bff)>(\Phi^{-1}(1-\alpha/2))^2$.

%\begin{remark}
%In Step 1, when the dimensionality of $\bff$ is large, one can replace the OLS estimator by the penalized least square estimator, which will be elaborated in Section \ref{subsec::test-high-dimension}.
%\end{remark}

Theoretical properties of the proposed conditional independence test will be studied in Section \ref{sec::TheoreticalP}.
In the above method, we implicitly assume that the number of variables $K$ is large so that the penalized least-squares methods are used.  When the number of variables $K$ is small, we can take $\lambda_1 = \lambda_2 = 0$ so that no penalization is imposed.

We would like to point out that after getting the estimated error matrices $\widehat\bE_x$ and $\widehat\bE_y$, one could apply other dependency measures including Hilbert Schmidt independence criterion \citep{gretton2005measuring}  and Heller-Heller-Gorfine test \citep{heller2012consistent}.

%Therefore, we construct test statistics $T(\beps_{x},\beps_{y},\alpha,n)$ in the following several steps with decision rule given in {\bf Step 5}.
%
%
%From Theorem \ref{thm::TestSigLev} in Section \ref{sec::TheoreticalP},  test $T(\beps_{x},\beps_{y},\alpha,n)$ with above rejection rule has an asymptotic significance level at most $\alpha$.

\subsection{Building graphs via conditional independence test}\label{subsec::graph-building}
Now we explore a specific application of our conditional independence test  to graphical models. To identify the conditional independence relationship in a graphical model, i.e.,  $z^{(i)}\independent z^{(j)}|\bz\setminus \{z^{(i)},z^{(j)}\}$, we assume
\begin{align}\label{eq::graph-internal-factor}
z^{(i)}_k=\bbeta^{\top}_{1,{ij}} \bff_k+\eps^{(i)}_{k},\ \ z^{(j)}_k=\bbeta^{\top}_{2,{ij}} \bff_k+\eps^{(j)}_{k},\ \ k=1,\dots,n,
\end{align}
where $\bff_k=(\bz^{(-i,-j)}_k)^{\top}$ represents all coordinates of $\bz_k$ other than $\bz^{(i)}_k$ and $\bz^{(j)}_k$, and $\bbeta_{1,{ij}}$ and $\bbeta_{2,{ij}}$ are $d-2$ dimensional  regression coefficients. Under model \eqref{eq::graph-internal-factor}, we decide whether edge $e_{ij}$ will be drawn through  directly testing $z^{(i)}\independent z^{(j)}|\mathcal{L}(\bz^{(-i, -j)})$, where $\mathcal{L}(\bff)$ is the linear space spanned by $\bff$.

%\noindent{\bf Remark:} Comparing with Section \ref{subsec::T}, $\bx$ and $\by$ are one-dimensional random variable currently and $K=p-2$.

More specifically, for each node pair $\{(i,j): 1\leq i<j \leq d\}$, we define $T^{(i,j)}=T(z^{(i)},z^{(j)},\bz^{(-i,-j)})$ using the same steps as in Section \ref{subsec::T} as the test for the current null hypothesis:
\begin{align}\label{eqn::edgetest}
  H_{0, ij}:  \beps^{(i)}\independent \beps^{(j)}. %\mbox{ vs. } H_{1, ij}:  \text{not $H_{0,ij}$}.
\end{align}

We now summarize the testing results by a graph in which nodes represent variables in $\bz$ and the edge $e_{ij}$ between node $i$ and node $j$ is drawn only when $H_{0, ij}$ is rejected at level $\alpha$.

In \eqref{eq::graph-internal-factor}, the factors are created internally via the observations on remaining nodes $\bz\setminus \{z^{(i)},z^{(j)}\}$. In financial applications, it is often desirable to build graphs when conditioning on external factors. In such cases, it is straightforward to change the factors in \eqref{eq::graph-internal-factor} to  external factors.

We will demonstrate the two different types of conditional dependency graphs via examples in Sections \ref{sec::NumericalS} and \ref{sec::Realdata}.

\subsection{Graph estimation with FDR control}\label{subsec::FDR control}
 Through the graph building process described in Section \ref{subsec::graph-building},  we can carry out $\bar{d}= d(d-1)/2$ P-DCov tests simultaneously and we wish to control the {\it false discovery rate} (FDR) at a pre-specified level $0<\alpha<1$. Let $R_F$ and $R$ be the number of falsely rejected hypotheses and the number of total rejections, respectively. The {\it false discovery proportion }(FDP) is defined as $R_F/\max\{1, R\}$ and the FDR is the expectation of FDP.

In the literature, various procedures have been proposed for conducting large-scale multiple hypothesis testing via FDR control.  \cite{liu2013gaussian} proposed a procedure for estimating large Gaussian graphical models with FDR control.  \cite{fan2018farm} proposed dependency-adjusted tests by estimating the latent factors that drive the dependency of these tests.
 In this work, we will follow the most commonly used Benjamini and Hochberg (BH) procedure developed in the seminal work of \cite{benjamini1995controlling}, where P-values of all marginal tests are compared. More specifically, let $P_{(1)}\leq P_{(2)} \leq \dots \leq P_{( \bar{d} )}$ be the ordered P-values of the $\bar{d}$ hypotheses given in (\ref{eqn::edgetest}). Let $s=\max\{ 0 \leq i \leq \bar{d} : P_{(i)}\leq  \alpha i /\bar{d} \}$, and we reject the $s$ hypotheses $H_{0, ij}$ with the smallest P-values. We will demonstrate the performance of this strategy via the real data example in Section \ref{sec::Realdata}.  

\subsection{Extension to functional projection}\label{subsec:: Functional Projection}
In the P-DCov described in Section \ref{subsec::T}, we assume the conditional dependency of $\bx$ and $\by$ given  factor $\bff$ is expressed via a linear form of $\bff$. In other words, we are projecting $\bx$ and $\by$ onto the space orthogonal  to $\mathcal{L}(\bff)$ and evaluate the dependence between the projected vectors. Although this linear projection assumption makes the theoretical development easier and delivers the main idea of this work, a natural extension is to consider a nonlinear projection. In particular, we consider the following additive generalization \citep{stone1985additive}  of the factor model setup:
\begin{align}\label{eq::model-additive}
\bx = \sum_{j=1}^K \bg^x_j(f_{j})+\beps_{x}, \by=  \sum_{j=1}^K \bg^y_j(f_{j})+\beps_{y},
\end{align}
where $\{\bg^x_j(\cdot), \bg^y_j(\cdot), j=1,\dots,K\}$ are unknown vector-valued functions we would like to estimate.   In \eqref{eq::model-additive}, we consider the additive space spanned by factor $\bff$. By this extension, we could identify more general conditional dependency structures between $\bx$ and $\by$ given $\bff$. This is a special case of \eqref{eq::model-general}, but avoids the issue of curse of dimensionality.

In the high-dimensional setup where $K$ is large, we can use a penalized additive model  \citep{ravikumar2009sparse, fan2011nonparametric}  to estimate the unknown functions. The conditional independence test described in Section \ref{subsec::T} could be modified by replacing the linear regression with the (penalized) additive model regression.  We will investigate the P-DCov method coupled with the sparse additive model \citep{ravikumar2009sparse} in  numerical studies.

\section{Theoretical Results}\label{sec::TheoreticalP}
In this section, we derive the asymptotic properties of our conditional independence test. First, we introduce  several assumptions  on $\beps_x$, $\beps_y$ and $\bff$.

\begin{condition}\label{Assump:Moments}
$\mathbb{E}\beps_{x}=\mathbb{E}\beps_{y}={\bf 0}$,\ \ $\mathbb{E} \|\beps_{x}\|^2<\infty$, $\mathbb{E} \|\beps_{y}\|^2<\infty$.

\end{condition}

\begin{condition}\label{Assump:Tail Condition} We denote  $h_x$ as the density function of random variable $x$. Let us assume that the densities of $\|\beps_{1,x}-\beps_{2,x}\|$ and $\|\beps_{1,y}-\beps_{2,y}\|$ are bounded on $[0,C_0]$, for some positive constant $C_0$. In other words, there exists a positive constant $M$,
\[
\max_{t\in[0,C_0]} h_{\|\beps_{i,x}-\beps_{j,x}\|}(t)\leq M, \ \ \ \ \max_{t\in[0,C_0]} h_{\|\beps_{i,y}-\beps_{j,y}\|}(t)\leq M.
\]
\end{condition}
 
\begin{remark}
	Conditions \ref{Assump:Moments} and \ref{Assump:Tail Condition} impose mild moment and distributional assumptions on  random errors $\beps_x$ and $\beps_y$. We use the following two simple examples to provide some intuitions regarding Condition \ref{Assump:Tail Condition}.  Assume $\beps_{i,x}\sim\mathcal{N}({\bf 0},{\bf I}_p)$, for $i=1,\ 2$, we have $\beps_{1,x}-\beps_{2,x}\sim\mathcal{N}({\bf 0},2 {\bf I}_p)$ and hence $\|\beps_{1,x}-\beps_{2,x}\|^2\sim 2 \chi^2(p)$. Therefore, 
%\[h_{\|\beps_{i,x}-\beps_{j,x}\|}(t) = \frac{1}{2^{(p/2)-1}\Gamma(p/2)}t^{p-1} e^{-t^2/2}\]
%
\[
h_{\|\beps_{i,x}-\beps_{j,x}\|}(t) = \frac{1}{2^{p-1}\Gamma(p/2)}t^{p-1} e^{-\frac{t^2}{4}}.
\]

It is easy to observe that, with $C_0=1$ and $M=1$, Condition 2 is satisfied. Now instead of an identity covariance matrix, let us consider the other extreme case with all coordinates copies or negative copies of one variable (the case where all correlations equal 1 or -1). Then $\|\beps_{1,x}-\beps_{2,x}\|^2\sim 2p\cdot \chi^2(1)$. 
Therefore, 
\[
h_{\|\beps_{i,x}-\beps_{j,x}\|}(t) = \frac{1}{\sqrt{p}\cdot\Gamma(1/2)}e^{-\frac{t^2}{4p}}.
\]
Again, with $C_0=1$ and $M=\frac{1}{\Gamma(1/2)}$, Condition 2 is satisfied. 
%Beyond multivariate Gaussian, Condition 2 can also be verified for heavier-tail distributions including the $t$ distribution. 
\end{remark}

To better understand when the proposed projection method works, we give the following high-level assumptions, whose justifications are noted below. %\begin{condition}\label{Assump:RSC}
%Suppose there exists $\gamma>0$, such that for all nonzero $v$ with $\|v_{S^c}	\|_1\le 3\|v_{S}\|_1$,
%\begin{align}
%	\frac{v^T\bff^T\bff v}{n\|v\|_2^2}\ge \gamma.
%\end{align}
%\end{condition}

\begin{condition}\label{Assump:dn} There exist   constants $C_1>1$ and $\gamma>0$, such that for any $C_2>1$, with probability greater than $1-C_1^{-C_2}$, we have for any $n$,
\[
 \|(\bB_x-\hbB_x)\bF\|_{2,\infty} \leq C_2 a_n,  \qquad \|(\bB_y-\hbB_y)\bF\|_{2,\infty} \leq C_2 a_n,
\]
where the sequence $a_n=o\{n^{-1/4}\wedge(n^{(1+\gamma)}\log n)^{-1/3}\}$.

%and the matrix
%
%\[
% \max_{i,j}\|(\bB_x-\hbB_x)(\bff_i-\bff_j)\|\leq C_2 a_n,
%\]
%\adc{The previous is equivalent to the following
%\[
% \max_{i}\|(\bB_x-\hbB_x)\bff_i\|\leq C_2 a_n,
%\]
%}

%\begin{align*}
%\max_{i,j}|(B_x-\hbB_x)(\bff_i-\bff_j)| &= O_p(a_n),\cr
%\max_{i,j}|(B_y-\hbB_y)(\bff_i-\bff_j)| &= O_p(a_n).
%\end{align*}
\end{condition}

\begin{condition}\label{Assump:l1}
Let $\bB_{x,l}$ denote the $l$-th row of $\bB_x$, and similarly we define $\hbB_{x,l}$, $\bB_{y,l}$ and $\hbB_{y,l}$.  We assume for any fixed $l$,
\begin{align*}
\|\bB_{x,l}-\hbB_{x,l}\|_1 = O_p(e_n), \qquad \|\bB_{y,l}-\hbB_{y,l}\|_1 = O_p(e_n),
\end{align*}
where sequences $e_n$ and $a_n$ in Condition \ref{Assump:dn} satisfy $a_n e_n = o(\frac{1}{\sqrt{n}\log K })$.
\end{condition}

\begin{remark}
Conditions \ref{Assump:dn} and \ref{Assump:l1} are mild. They are imposed to ensure the quality of the projection and guarantee the theoretical properties regarding our conditional independence test. For example, one could directly call the results from penalized least squares for high-dimensional regression \citep{belloni2011square, buhlmann2011statistics, hastie2015statistical}  and  robust estimation \citep{belloni2011, wang2013l1,  fan2014robust}. We now discuss two special examples as follows.
\begin{enumerate}
\item ($K$ is fixed) In this fixed dimensional case, it is straightforward to verify that the projection based on ordinary least squares satisfies the two conditions.
\item (Sparse Linear Projection) Let $\bB_x = (\bb_1^T, \bb_2^T,\dots, \bb_p^T)^T$ and $\widehat\bB_x = (\widehat\bb_1^T, \widehat\bb_2^T,\dots, \widehat\bb_p^T)^T$. Note that the graphical model case corresponds to  $p=1$.   We apply the popular $L_1$-regularized least squares for each dimension of $\bx$ regressing on the factor $\bF$. Here, we further assume the true regression coefficient $\bb_j$ is sparse for each $j$ with $S_j = \{k: (\bb_j)_k\neq 0\}$, $\hat S_j = \{k: (\widehat\bb_j)_k\neq 0\}$ and $|S_j|=s_j$. From Theorem 11.1, Example 11.1 and Theorem 11.3 in \cite{hastie2015statistical}, and since  $\{\bff_i\}_{i=1}^n$ are i.i.d., we have with probability going to 1,
$\|\widehat \bb_j -\bb_j\|\le C\sqrt{\frac{s_j\log K}{n}}$, $\hat S_j = S_j$ and $\max_i \|(\bff_i)_{S_j}\|\le s_j \log n$. Then, we have with probability going to 1, for each $i=1,\dots, n$ and $j = 1,\dots,p$,
\begin{align}\label{eq::lasso.case}
	\|(\widehat \bb_j -\bb_j)^T\bff_i\|  =  \|(\widehat \bb_j - \bb_j)_{S_j}^T(\bff_i)_{S_j}\| &\le  \|(\widehat \bb_j - \bb_j)_{S_j}\|\|(\bff_i)_{S_j}\|\nonumber\\
	&\le Cs_{\max} \log n \sqrt{\frac{s_{\max}\log K}{n}},
\end{align}
where $s_{\max}=\max_j s_j$.
It is now easy to verify that Condition \ref{Assump:dn} and \ref{Assump:l1} are satisfied even under the ultra-high-dimensional case where $\log K=o(n^a), 0<a<1/3$. We would like to omit the details here for brevity about the specification of various constants.
\end{enumerate}
\end{remark}

%Let $\bB_x = (\bb_1^T, \bb_2^T,\dots, \bb_p^T)^T$ and $\widehat\bB_x = (\widehat\bb_1^T, \widehat\bb_2^T,\dots, \widehat\bb_p^T)^T$ . When using the $L_1$ penalty, we have the  the estimation error $\|\hat \bb_j -\bb_j\|$ upper bound and the variable selection consistency of $\hat\bb_j$ from Theorems 11.1 and 11.3 in \cite{hastie2015statistical}, respectively. Then, under the associated regularity conditions (with the details skipped for simplicity), we have with high probability,
%\begin{align}
%	\|(\hat \bb_j -\bb_j)^T\bff_i\|  = \|(\hat \bb_j - \bb_j)_S^T(\bff_i)_S\| = \|(\hat \bb_j - \bb_j)_S\|\|(\bff_i)_S\|.
%\end{align}
%
%\adc{
%From Corollary 6.1 in \cite{buhlmann2011statistics}, we have under certain regularity conditions, with high probability,
%\begin{align}
%	\frac{|(\hat b_j - b_j)F|^2}{n} \le C \sqrt{\frac{\log K}{n}},
%\end{align}
%for $j=1,\dots, p$, where $B_x = (b_1^T, b_2^T,\dots, b_p^T)^T$.
%However, our condition 3 needs $|(\hat b_j - b_j)^Tf_i|^2 $ to be small, which is not true from the $L_2$ norm of the vector as there is a factor of $n$ when considering the maximum.
%Use $|(\hat b_j - b_j)^Tf_i|^2  = \|(\hat b_j - b_j)_S^T(f_i)_S\| = \|(\hat b_j - b_j)_S\|\|(f_i)_S\|$  Then $\|(f_i)_S\|=O(s)$. $\max_i \|(f_i)_S\| = O_p((\log n) s)$.
%}
%\begin{prop}
%	
%\end{prop}

\begin{theorem}\label{thm::consistency}Under Conditions \ref{Assump:Moments} and \ref{Assump:dn},
\[
\mathcal{V}^2_n(\hbeps_x,\hbeps_y)\convinprob\mathcal{V}^2(\beps_x,\beps_y).
\]

In particular, when $\beps_x$ and $\beps_y$ are independent, $\mathcal{V}^2_n(\hbeps_x,\hbeps_y)\convinprob 0$.
\end{theorem}
%
%
%\begin{lemma}[Lemma in \cite{buhlmann2011statistics}] \label{lemma:lasso-results}
%Under the regularity conditions, we have $P[\mathcal{E}_n]>1-...$.
%\end{lemma}

Theorem \ref{thm::consistency} shows that the sample distance covariance between the estimated residual vectors converges to the distance covariance between the population error vectors. It enables us to use the distance covariance of the estimated residual vectors to construct the conditional independence test as described in Section \ref{subsec::T}.

\begin{theorem}\label{thm::null-distribution} Under Conditions \ref{Assump:Moments}-\ref{Assump:l1}, and  the null hypothesis that $\beps_x \independent \beps_y$ (or equivalently $\bx\independent \by|\bff$),
$$
n\mathcal{V}^2_n(\hbeps_x,\hbeps_y)\convindist\|\zeta\|^2,
$$
where $\zeta$ is a zero-mean Gaussian process defined analogously as in Lemma~\ref{thm::old-asymptotic}.
\end{theorem}

Theorem \ref{thm::null-distribution} provides the asymptotic distribution of the test statistic $T(\bx,\by,\bff)$ under the null hypothesis, which is the basis of Theorem \ref{thm::TestSigLev}.

\begin{corollary}\label{coro::new-teststat}
Under the same conditions of Theorem \ref{thm::null-distribution}, \[
n\mathcal{V}^2_n(\hbeps_x,\hbeps_y) /S_2(\hbeps_x,\hbeps_y) \convindist Q, \ \ \text{where}\ \ Q\overset{\mathcal{D}}=\sum_{j=1}^\infty \lambda_j Z^2_j,
\]
where $Z_j\overset{i.i.d}{\sim}\mathcal{N}(0,1)$ and $\{
    \lambda_j\}$ are non-negative constants depending on the distribution of $(\bx,\by)$; $\mathbb{E}(Q)=1$.
\end{corollary}

\begin{theorem}\label{thm::TestSigLev}
Consider the test that rejects conditional independence when
\begin{align}\label{eq::reject-region}
\frac{n\mathcal{V}^2_n(\hbeps_{x},\hbeps_{y})}{S_2(\hbeps_x,\hbeps_y)}>(\Phi^{-1}(1-\alpha/2))^2,
\end{align}
where $\Phi(\cdot)$ is the cumulative distribution function of $\mathcal{N}(0,1)$. Let $\alpha_n(\bx,\by,\bff)$ denote its associated type I error. Then under Conditions \ref{Assump:Moments}-\ref{Assump:l1}, for all $0<\alpha\le0.215$,
\begin{enumerate}[label=(\roman*)]
\item $\lim_{n\rightarrow\infty}\alpha_n(\bx,\by,\bff)\leq \alpha$,
\item $\sup_{\beps_{x}\independent\beps_{y}}\lim_{n\rightarrow \infty}\alpha_n(\bx,\by,\bff)=\alpha$.
\end{enumerate}
\end{theorem}
Part (i) of Theorem \ref{thm::TestSigLev} indicates the proposed test with critical region \eqref{eq::reject-region} has an asymptotic significance error at most $\alpha$. Part (ii) of Theorem \ref{thm::TestSigLev} implies that there exists a pair $(\beps_x,\beps_y)$ such that the pre-specified significant level $\alpha$ is achieved asymptotically.  In other words, the size of testing $H_0: \beps_{x}\independent\beps_{y}$ is $\alpha$.

\begin{remark}
	When the sample size $n$ is small, the theoretical critical value in \eqref{eq::reject-region} could sometimes be too conservative in practice \citep{szekely2007measuring}. Therefore, we recommend using random permutation to get a reference distribution for the test statistic $T(\bx,\by,\bff)$ under $H_{0}$. Random permutation is used to decouple $\beps_{i,x}$ and $\beps_{i,y}$ so that the resulting pair  $(\beps_{\pi(i),x}, \beps_{i,y})$ follows the null model, where $\{\pi(1), \dots, \pi(n)\}$ are a random permutation of  indices $\{1, \dots, n\}$. Here, we set the number of permutations  $R(n)=\lfloor 200+5000/n\rfloor$ as in \cite{szekely2007measuring}.
	Consequently, we can also estimate the P-value associated with the conditional independence test based on the quantiles of the test statistics over $R(n)$ random permutations.
\end{remark}

\section{Monte Carlo Experiments}\label{sec::NumericalS}
In this section, we investigate the performance of P-DCov with five simulation examples.  In Example \ref{exmp::SimTestB}, we consider a factor model and test the conditional independence between two vectors $\bx$ and $\by$ given their common factor $\bff$, via P-DCov. In Examples \ref{exmp::GaussianGraph}, we investigate the classical Gaussian graphical model. In Example \ref{exmp::smu}, we consider the case of general graphical model without the Gaussian assumption. In Example \ref{exmp::factor-graph}, we consider the case of dependency graph with the contribution of external factors. In Example \ref{exmp::smu-factor}, we consider a general graphical model with external factors.

\begin{example}\label{exmp::SimTestB} [High-dimensional factor model]
Let $p=5$, $q=10$ and $K=1000$. The rows of $\bB_x$ and rows of $\bB_y$ are drawn independently from  $\bz_K=(\bz_1^T,\bz_2^T)^T$, where $\bz_1$ is a 3-dimensional vector with elements i.i.d. from $Unif[2,3]$ and $\bz_2=\bzero_{K-3}$.
 $\{\bff_i\}^n_{i=1}$ are i.i.d. from $\mathcal{N}({\bf 0},{\bf I}_K)$. We generate $n$ $i.i.d.$ copies $\{\br_i\}_{i=1}^n$ from  log-normal distribution $\ln \mathcal{N}({\bf 0},\bSigma)$ (heavy-tail)  where $\bSigma$ is an equal correlation matrix of size $(p+q)\times(p+q)$ with $\Sigma_{jk}=\rho$ when $j \not = k$ and $\Sigma_{jj}=1$.  $\beps_{i,x}$ and $\beps_{i,y}$ are the centered version of the first $p$ coordinates and the last $q$ coordinates of $\br_i$. Then, $\{\bx_i\}^n_{i=1}$ and $\{\by_i\}^n_{i=1}$ are generated according to $\bx_i=\bB_x \bff_i+\beps_{i,x}$ and $\by_i = \bB_y\bff_i+\beps_{i,y}$ correspondingly.
\end{example}

In Example \ref{exmp::SimTestB}, we consider a high-dimensional factor model with sparse structure. Note that the errors are generated from a heavy tail distribution to demonstrate the proposed test works beyond  Gaussian errors. We assume each coordinate of $\bx$ and $\by$ only depends on the first three factors.
We calculate $T(\bx,\by,\bff)$ in the P-DCov test, and $T_0 (\bx,\by,\bff)$ in which we replace $\hat{\beps}_{i,x}$ and $\hat{\beps}_{i,y}$ by the true $\beps_{i,x}$ and $\beps_{i,y}$ as an oracle test to compare with. To get reference distributions of $T(\bx,\by,\bff)$ and $T_0 (\bx,\by,\bff)$, we follow the permutation procedure as described in Section \ref{sec::TheoreticalP}.  In this example, we set the significance level $\alpha= 0.1$.
We vary the sample size from 100 to 300 with increment of 20 and show the empirical power based on 2000 repetitions for both $T(\bx,\by,\bff)$ and $T_0(\bx,\by,\bff)$ in Figure \ref{fig:: Example 1} for $\rho\in\{0.1,0.2,0.3,0.4\}$. In the implementation of  penalized least squares in Step 1, we use R package \pkg{glmnet} with the default tuning parameter selection method (10-fold cross-validation) and perform least square on the selected variables  to reduce estimation bias of these estimated parameters \citep{belloni2013least}. It is worth mentioning that an alternative approach to reduce the estimation bias is the de-biased lasso method \citep{zhang2014confidence,van2014asymptotically}. Here, we decided to use the least square post model selection approach due to its simplicity and computational efficiency.

From Figure \ref{fig:: Example 1}, it is clear that as the sample size or $\rho$ increases, the empirical power also increases in general. Also, comparing the panels (A) and (B) in Figure \ref{fig:: Example 1}, we  see that when the sample size is small, the P-DCov test has smaller power than the oracle test, however, the difference between them becomes negligible as the sample size increases.  This is consistent with our theory regarding the asymptotic distribution of the test statistics.  When $\rho=0$, Table \ref{Table::Example 1} reports the empirical type I error for both P-DCov as well as the oracle version. It is clear that the type I error of P-DCov is under good control as the sample size increases.

\begin{table}
\caption{Type I error of Example 1\label{Table::Example 1}}
\centering
\begin{tabular}{cccccccccccc}
\hline
& \multicolumn{11}{c}{Test based on $\hat{\beps}_x$ and $\hat{\beps}_y$}\\
\hline
$n$&100&120&140&160&180&200&220&240&260&280&300 \\ \hline
 & 0.119&0.114&0.116&0.100&0.098&0.097&0.092&0.102&0.094&0.091&0.096 \\
\hline
& \multicolumn{11}{c}{Test based on $\beps_x$ and $\beps_y$} \\
\hline
$n$&100&120&140&160&180&200&220&240&260&280&300 \\
\hline
&0.086&0.102&0.104&0.094&0.092&0.091&0.096&0.103&0.098&0.092&0.095
\\

\end{tabular}
\end{table}

\begin{figure}[tb!]
\caption{Power-sample size graph of Example 1\label{fig:: Example 1}
}
\begin{tabular}{cc}
	\includegraphics[scale=0.42]{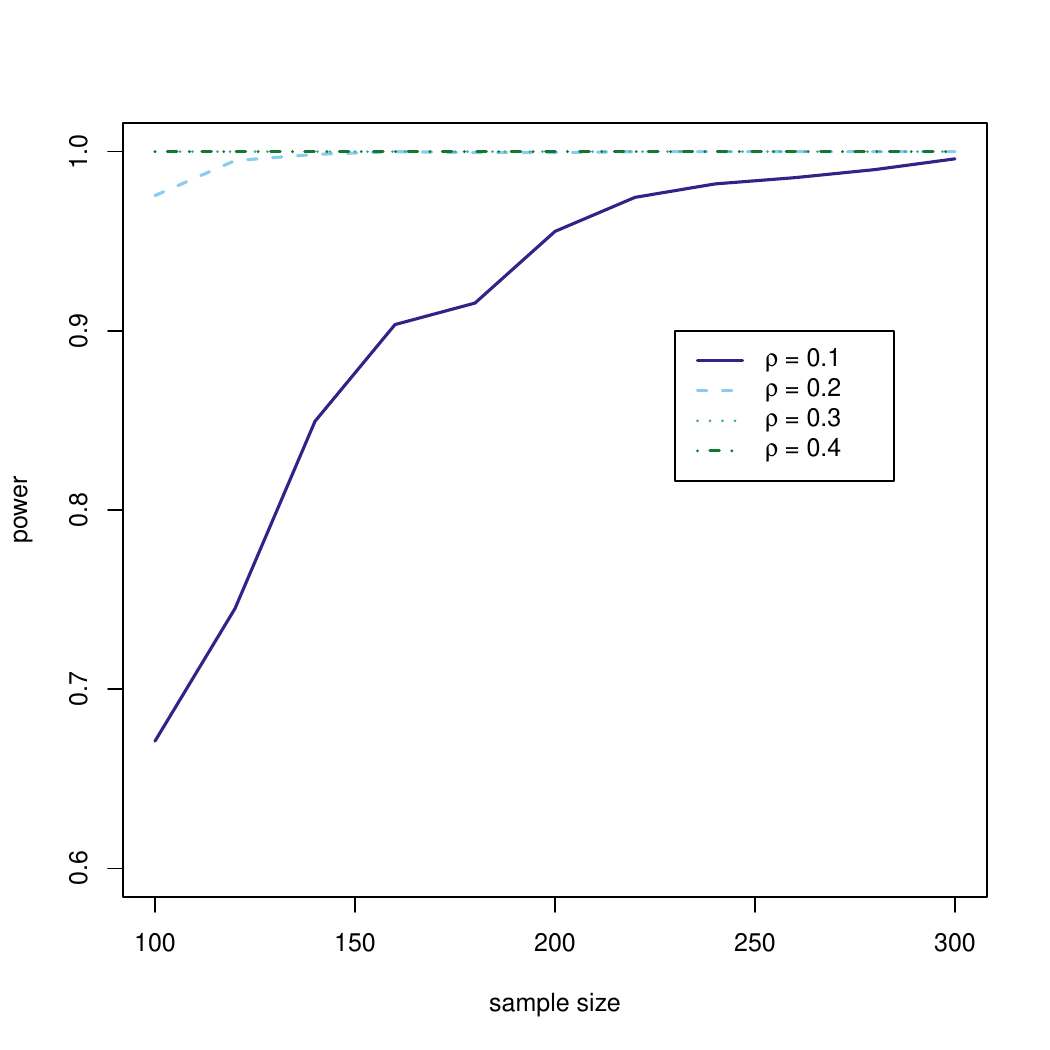}&\includegraphics[scale=0.42]{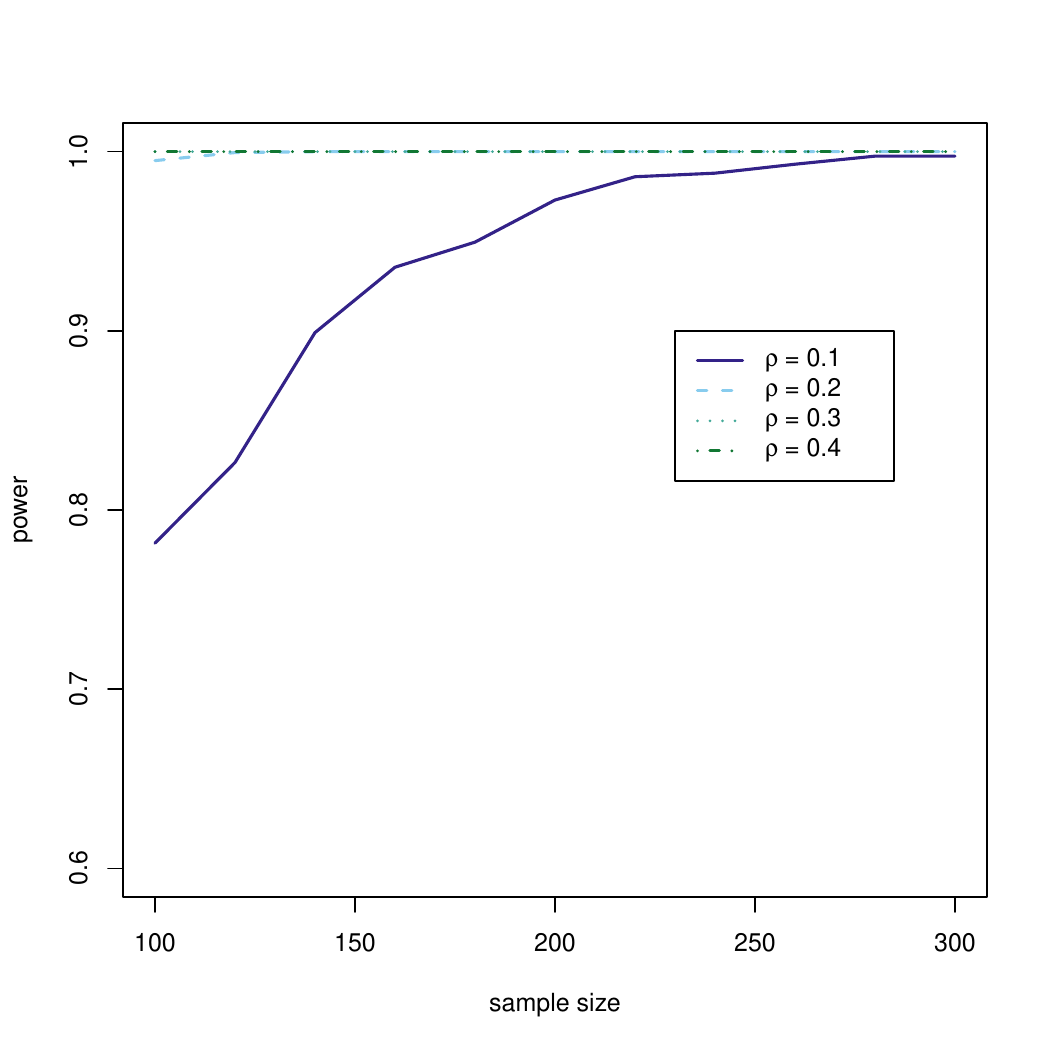}\\
	(A) $T(\bx,\by,\bff)$&(B) $T_0(\bx,\by,\bff)$\\
\end{tabular}

\end{figure}

\begin{example}\label{exmp::GaussianGraph}[Gaussian graphical model]
We consider a Gaussian graphical model with precision matrix $\bOmega=\bSigma^{-1}$,  where $\bOmega$ is a tridiagonal matrix of size $d\times d$, and is associated with the autoregressive process of order one. We set $d=100$ and the $(i,j)$-element in $\bSigma$ to be $\sigma_{i,j}=\exp(- | s_i-s_j |)$, where $0=s_1<s_2<\dots<s_d$. In addition,
\[
s_i-s_{i-1}\overset{\text{i.i.d}}{\sim}  \text{Uniform} \ (1,3), \ \ i=2,\dots,d.
\]

\end{example}

In this example, we would like to compare the proposed P-DCov with the state-of-the-art approaches for recovering Gaussian graphical models.
In terms of recovering structure $\Omega$, we compare lasso.dcov (projection by lasso followed by distance covariance), sam.dcov (projection by sparse additive model followed by distance covariance), lasso.pearson (projection by lasso followed by Pearson correlation), sam.pearson (projection by sparse additive model followed by Pearson correlation)  with three popular estimators corresponding to the lasso, adaptive lasso and scad penalized likelihoods (called  graphical.lasso, graphical.alasso and graphical.scad on the graph) for the precision matrix \citep{friedman2008sparse, fan2009network}. Here, lasso.dcov and sam.dcov are two examples of our P-DCov methods.  We use  R package \pkg{SAM} to fit the sparse additive model.  To evaluate the performances, we construct receiver operating characteristic (ROC) curves for each method with sample sizes $n=100$ and $n=300$. The process of constructing the ROC curves involves conducting the P-DCov test for each pair of nodes and record the corresponding P-values.  In each of the ROC curve, true positive rates (TPR) are plotted against false positive rates (FPR) at various thresholds  of those P-values (``TP" means the true entry of the precision matrix is nonzero and estimated as nonzero; ``FP" means the true entry of the precision matrix is zero but estimated as nonzero).  We follow the implementation in  \cite{fan2009network} for the three penalized likelihood estimators.  The average results over 100 replications of different methods are reported in Figure \ref{fig:: ggm}. The associated AUC (Area Under the Curve) for each method is also displayed in the legend of the figure.

\begin{figure}[tb!]
\caption{ROC curves for Gaussian graphical models with AUCs in legends.\label{fig:: ggm}}
\begin{tabular}{cc}
	\includegraphics[scale=0.4]{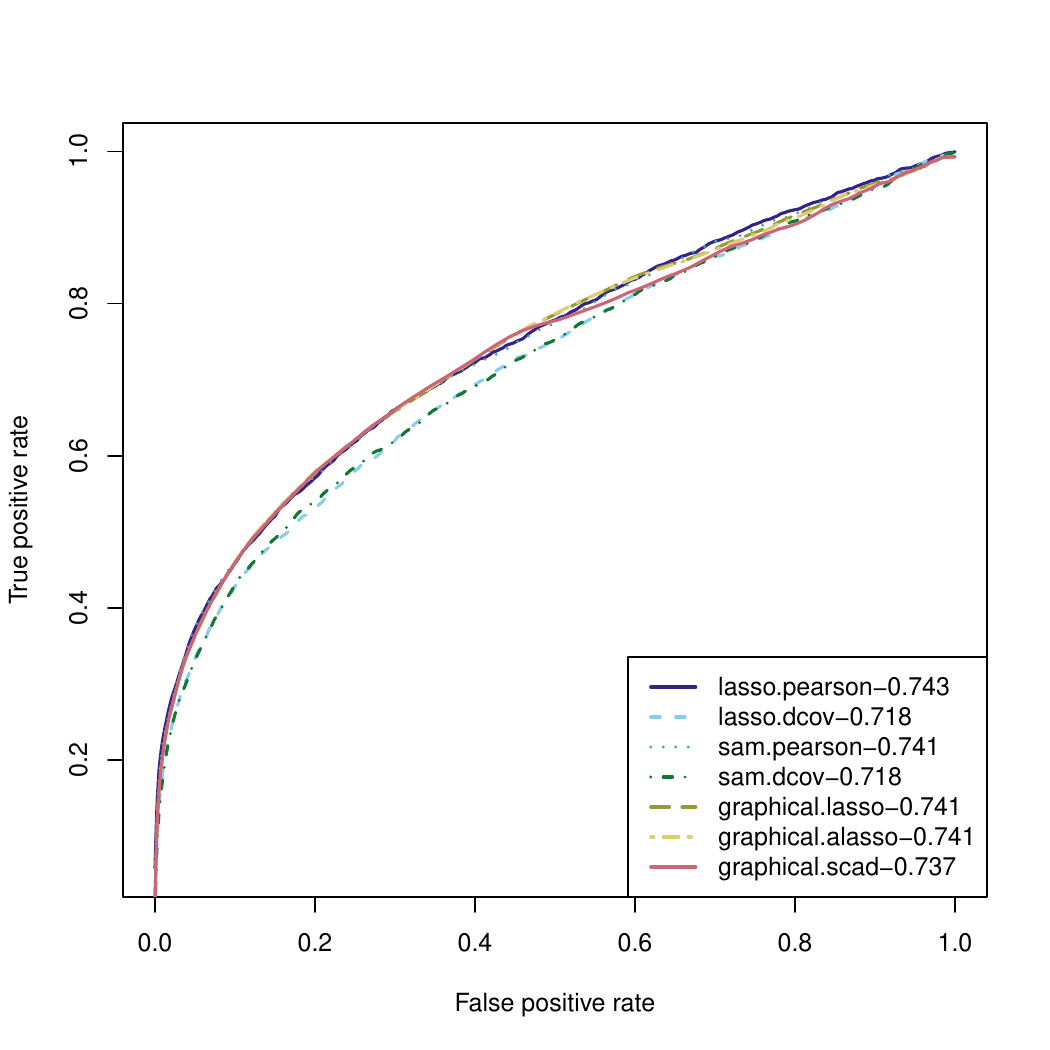}&\includegraphics[scale=0.4]{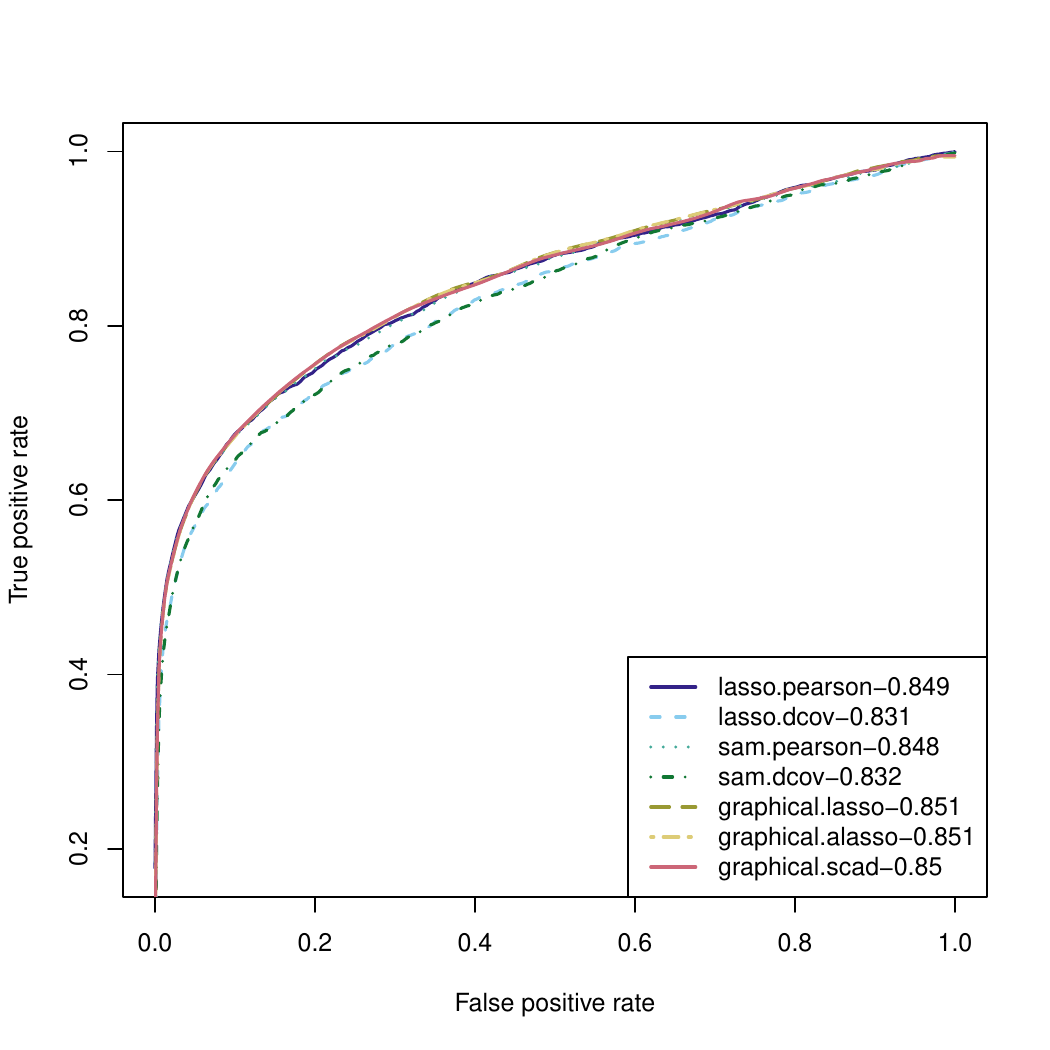}\\
	(A) $n=100$&(B) $n=300$\\
\end{tabular}

\end{figure}

We observe that lasso.pearson and sam.pearson perform similarly to the penalized likelihood methods when $n=100$. On the other hand, lasso.dcov and sam.dcov lead to slightly smaller AUC value due to the use of the distance covariance, which is expected for the Gaussian model.  This shows that we do not pay a big price for using the more complicated distance covariance and sparse additive model.
%Surprisingly, both the lasso.pearson and sam.pearson perform better than all the likelihood based methods when $n=300$, which shows the projection based approach for creating a dependency graph is very competitive.

\begin{example}\label{exmp::smu}[A general graphical model]
We consider a general graphical model with a combination of multivariate $t$ distribution  and  multivariate Gaussian distribution. The dimension of $\bx$ is $d=100$.  In detail,
$\bx= (\bx_1^T,\bx_2^T,\bx_3^T)^T$ where $\bx_1$ follows a 20 dimensional multivariate $t$ distribution with degrees of freedom 5, location parameter $0$ and identity covariance matrix, $\bx_2$  follows the same Gaussian graphical model as in Example \ref{exmp::GaussianGraph} except the dimension is now 10, and $\bx_3\sim\mathcal{N}({\bf 0}, {\bf I}_{70})$. In addition, $\bx_1$, $\bx_2$, and $\bx_3$ are mutually independent.
\end{example}
To generate a multivariate $t$-distribution, we first generate a random vector $\bw_{20}$ from the standard multivariate Gaussian distribution and  an independent random variable $\tau \sim \chi^2 (5)$ and then set $\bx_1 = \bw/\sqrt{\tau}$. One important fact about the multivariate $t$ distribution is that the zero element in the precision matrix does not imply conditional independence like the case of Gaussian graphical models \citep{finegold2009robust}. Indeed, for $\bx_1$, we actually have the fact that $\bx_1^{(i)}$ and $\bx_1^{(j)}$ are dependent given $\bx_1^{(-i, -j)}$ for any pair $1\le i\neq j\le 20$. On the contrary, the Gaussian likelihood based methods will falsely claim that all the components of $\bx_1$ are independent, because the corresponding elements in $\bOmega$ are 0.

The average ROC curve results are rendered in Figure \ref{fig:: smu}. As expected, by using the new projection based distance covariance method for testing conditional independence, lasso.dcov outperforms all the other methods in terms of AUC, with a more evident advantage when $n=300$. One interesting observation is that: in the region where FPR is very low, the likelihood based methods actually outperform P-DCov methods. One possible reason is that the likelihood based methods are more capable of capturing the conditional dependency structure within $\bx_2$ as it follows a Gaussian graphical model.

\begin{figure}[tb!]
\caption{ROC curves for a general graphical model  with AUCs in legends.\label{fig:: smu}}
\begin{tabular}{cc}
	\includegraphics[scale=0.4]{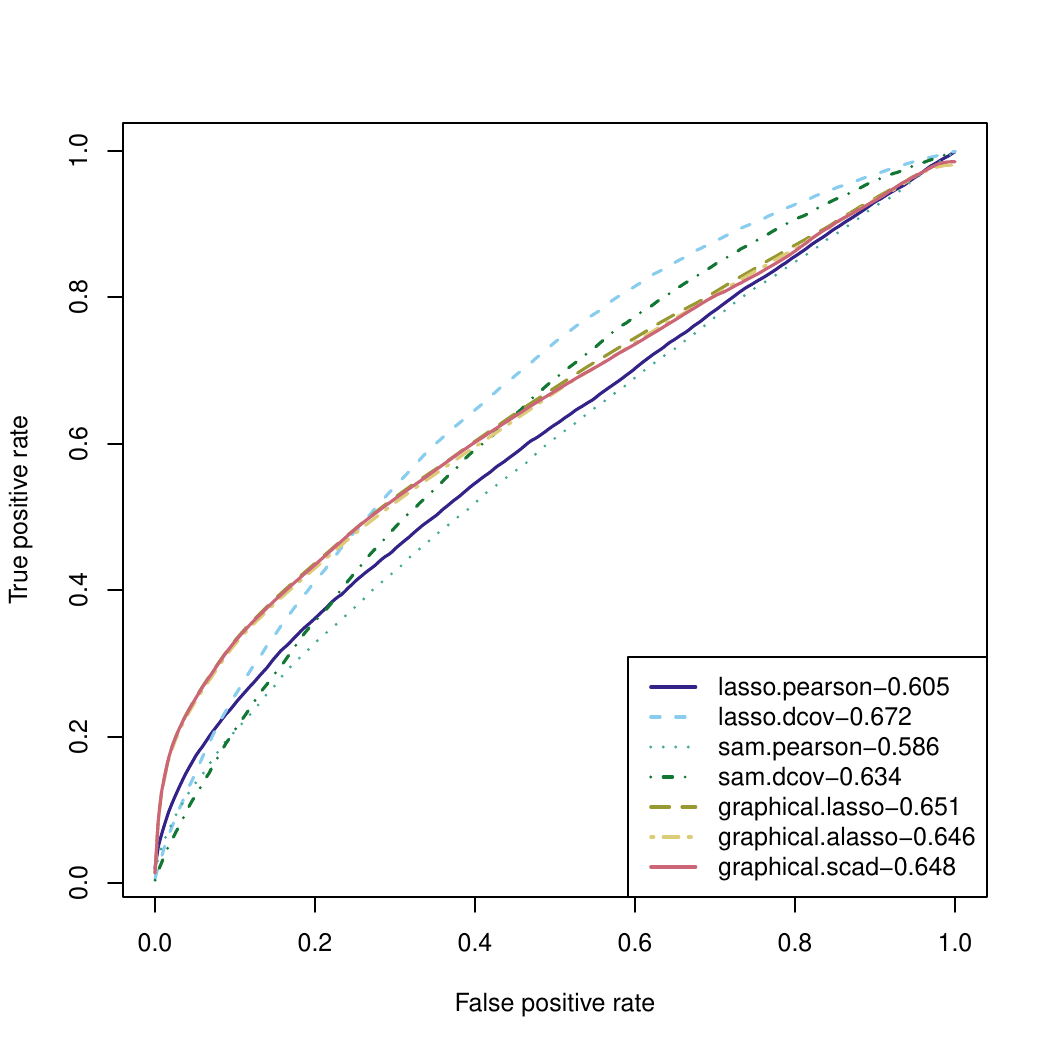}&\includegraphics[scale=0.4]{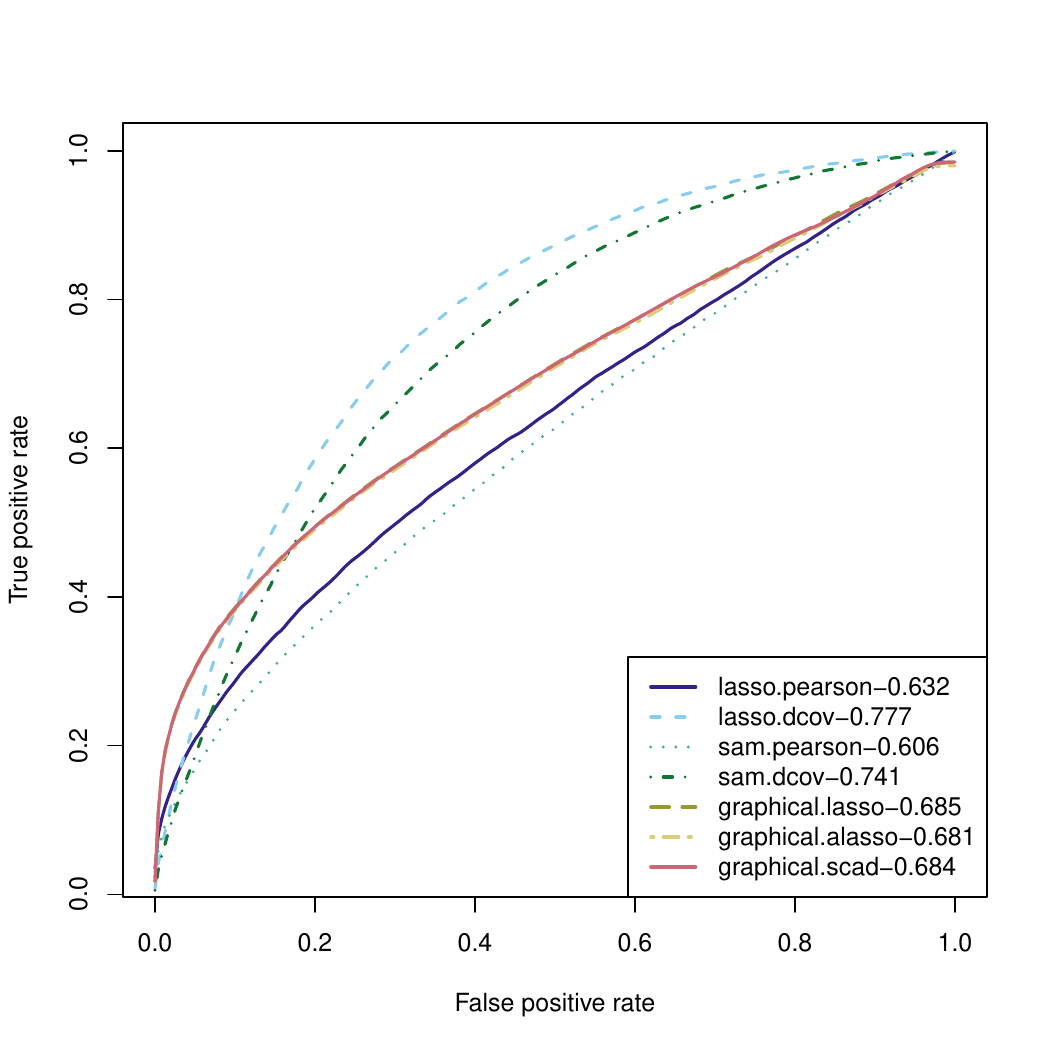}\\
	(A) $n=100$&(B) $n=300$\\
\end{tabular}

\end{figure}

%
%\begin{figure}[tb!]
%\begin{subfigure}[b]{0.5 \textwidth}
%\includegraphics[scale=0.47]{Figures/smu-roc-n=100.pdf}
%\caption{$n=100$}
%\end{subfigure}
%\begin{subfigure}[b]{0.5 \textwidth}
%\includegraphics[scale=0.47]{Figures/smu-roc-n=300.pdf}
%\caption{$n=300$}
%\end{subfigure}
%\caption{ROC curves for a general graphical model}
%\label{fig:: smu}
%\end{figure}
%
%
\begin{example}\label{exmp::factor-graph}[Dependency graph with external factors]
We consider a dependency graph with the contribution of external factors. In particular, we generate $\bu\sim \mathcal{N}({\bf 0},\bOmega)$, where $\bOmega$ is the same tridiagonal matrix used in Example \ref{exmp::GaussianGraph} except the dimension is now 30 and $\bff \sim \mathcal{N}({\bf 0}, \bI_{300})$, then the observation $\bx = \bu + \bQ g(\bff)$ where $\bQ_{30\times 300}$ is a sparse coefficient matrix that dictates how each dimension of $\bx$ depends on the factor $g(\bff)$. In particular, we let $\bQ=[\tilde\bQ_{30\times 15}, {\bf 0}_{30\times 285}]$ with the generation of $\tilde\bQ$ follows the setting in \cite{cai2013covariate}.  For each element $\tilde Q_{ij}$, we first generate a Bernoulli distribution with success probability 0.2 to determine whether $\tilde Q_{ij}$ is 0 or not. If $\tilde Q_{ij}$ is not 0, we then generate $\tilde Q_{ij} \sim \mbox{Uniform}\mbox{ }(0.5,1)$. Here we consider two forms of $g(\cdot)$, namely $g(\bff)=\bff$ and $g(\bff)=\bff^2$.
\end{example}
Now, we report results regarding the average ROC curves  for lasso.pearson, lasso.dcov, \\sam.pearson and sam.dcov. The results for both $g(\bff)=\bff$ and $g(\bff)=\bff^2$ are depicted in Figure \ref{fig:: factor}. Note that we are not building a conditional dependency graph among $\bx$, but a dependency graph of $\bx$ conditioning on the external factor $\bff$. There are some insightful observations from the figure. First of all, by looking at the first case when $g(\bff) =\bff$, it is clear that lasso.pearson is the best as it takes advantage of the sparse linear structure paired with the Gaussian distribution of the residual. By using the distance covariance as a dependency measure, or by using the sparse additive model as a projection method, it is reassuring that we do not lose much efficiency. Second, for the case when $g(\bff)=\bff^2$ and $n=300$, we can see a substantial advantage of the sparse additive model based methods as they can capture this nonlinear contribution of the factors to the dependency structure of $\bx$.

\begin{figure}[tb!]
\caption{ROC curves for factor based dependency graph with AUCs in legends. \label{fig:: factor}}
\begin{tabular}{cc}
	\includegraphics[scale=0.4]{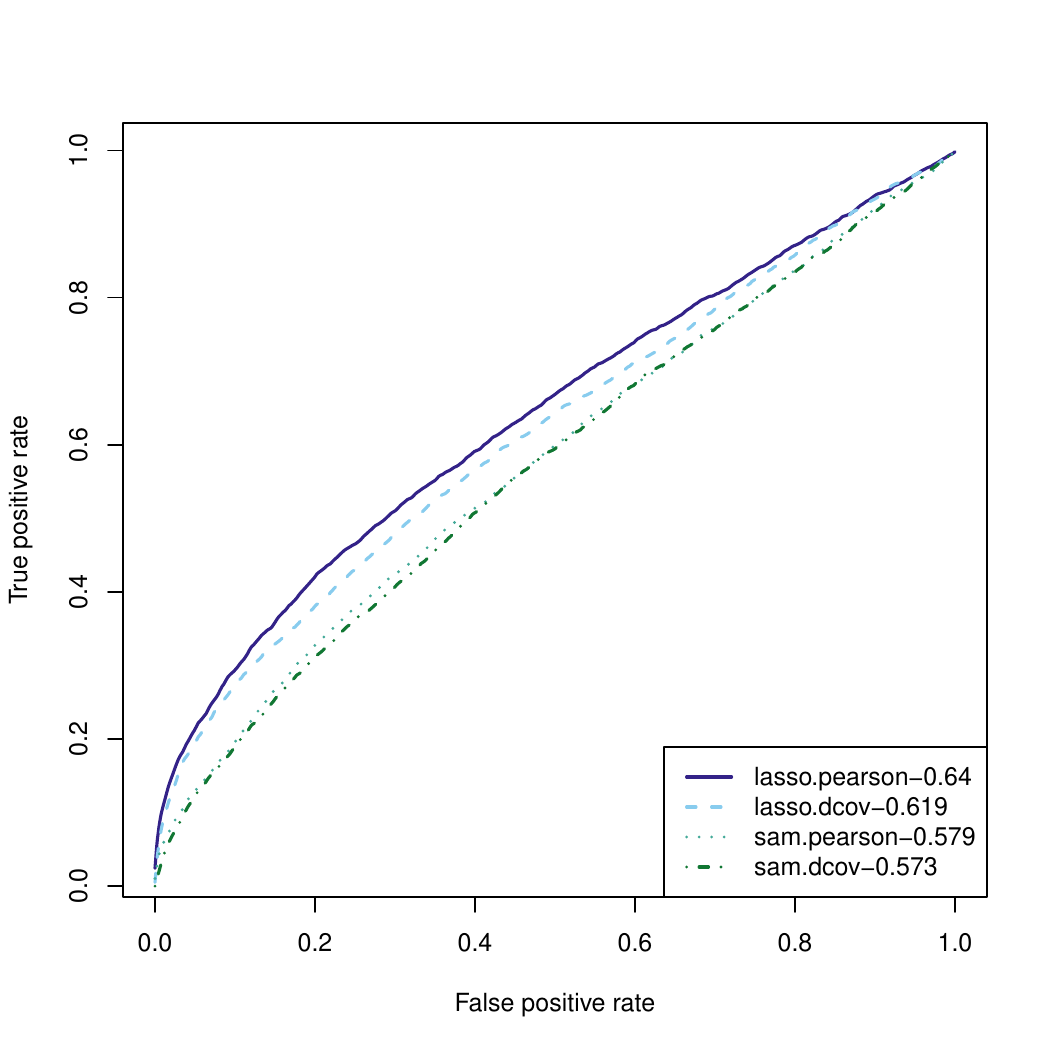}&\includegraphics[scale=0.4]{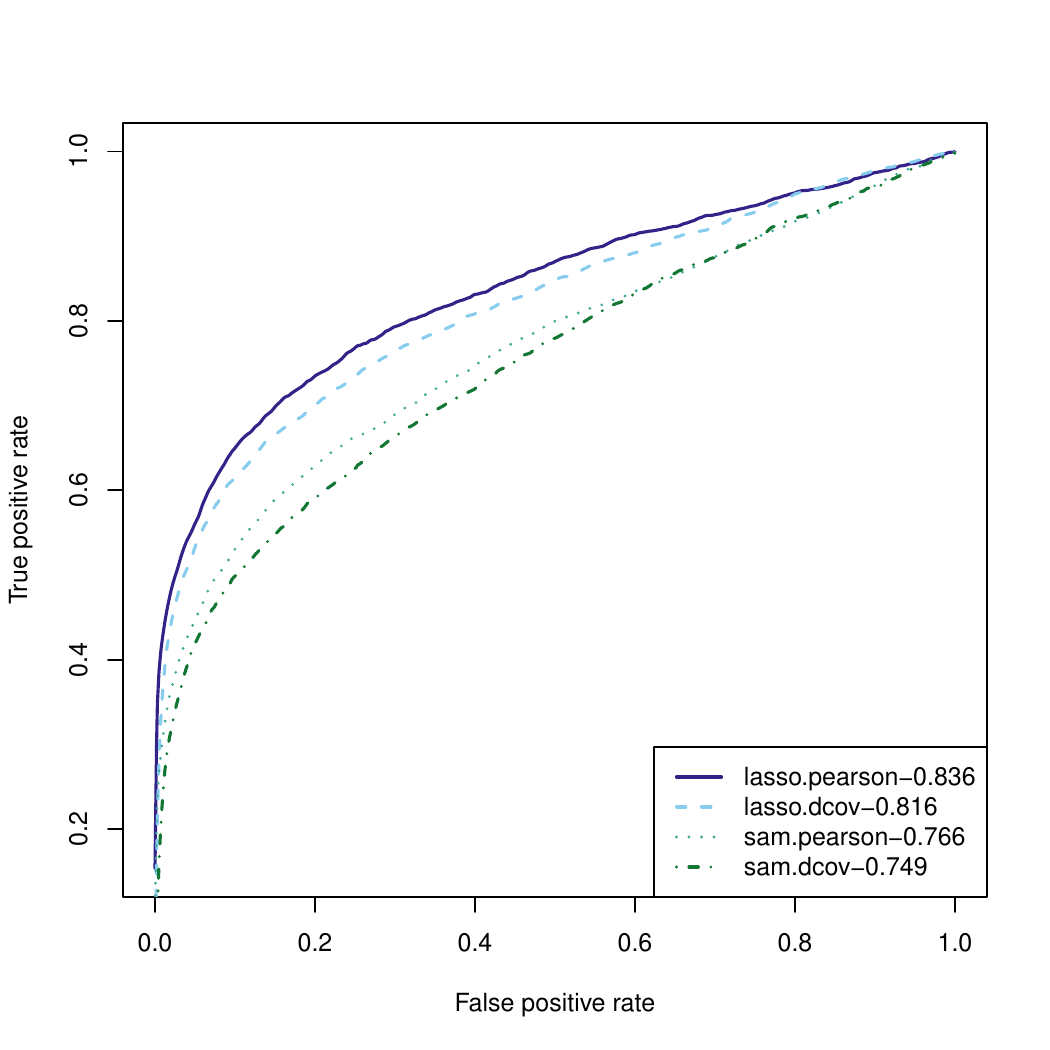}\\
	(A) $n=100$, $g(\bff)=\bff$ &(B) $n=300$, $g(\bff)=\bff$\\
\includegraphics[scale=0.4]{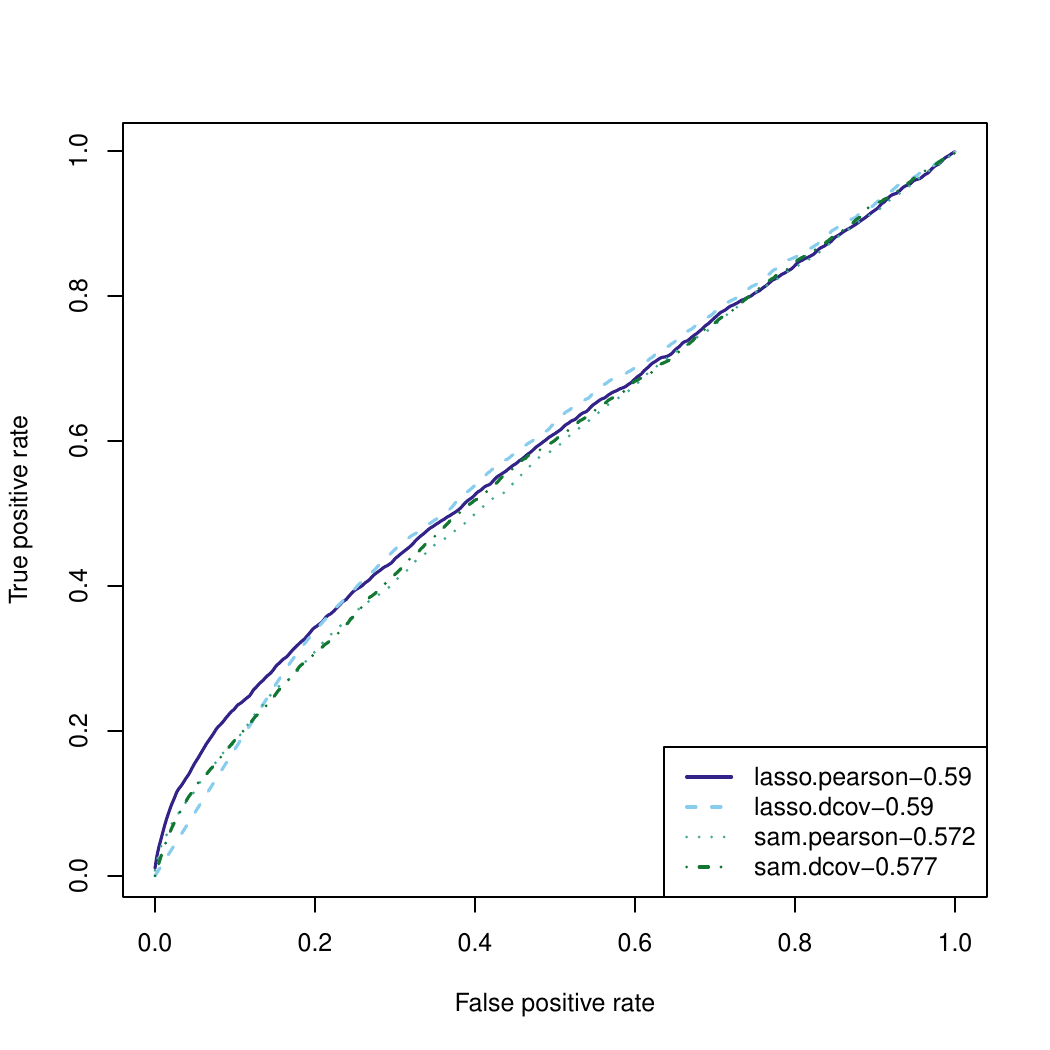}&\includegraphics[scale=0.4]{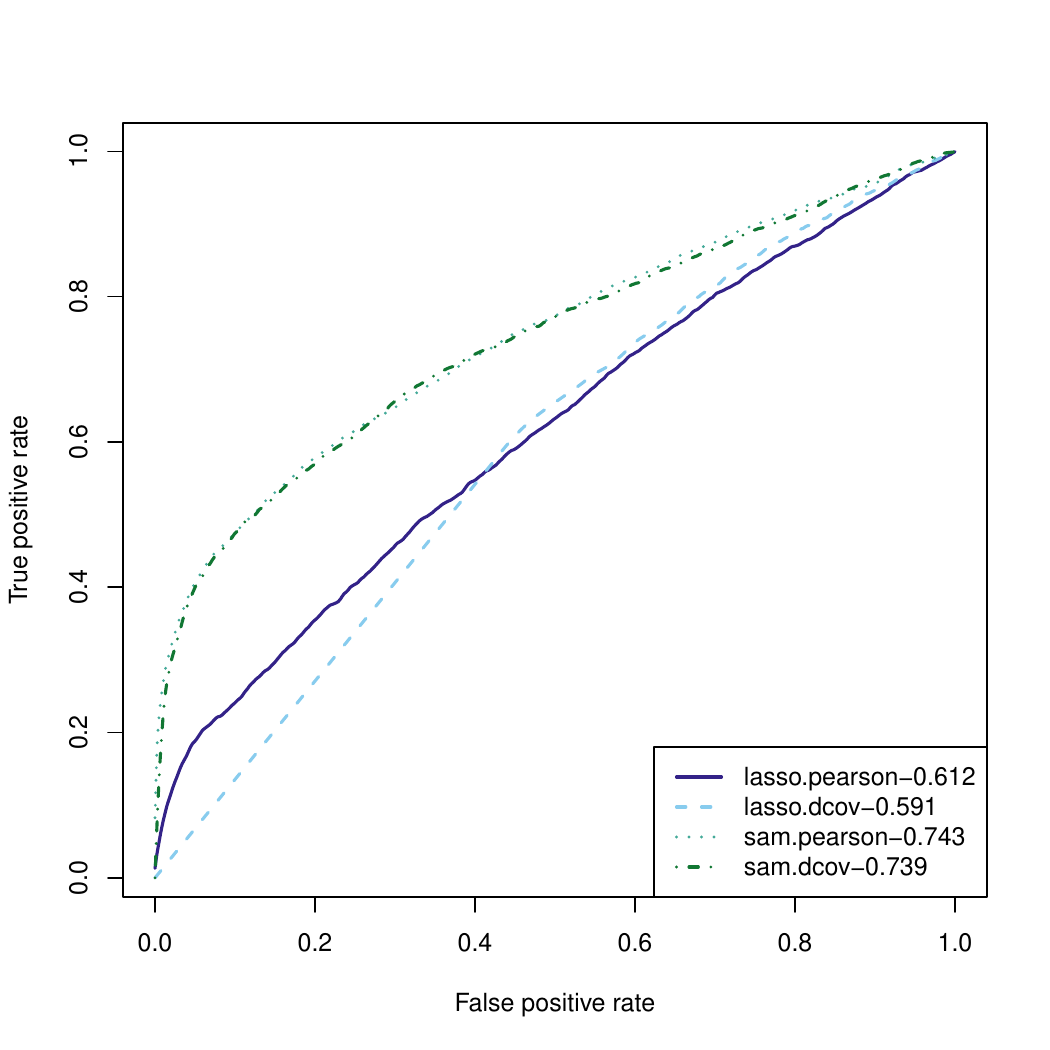}\\
	(C) $n=100$, $g(\bff)=\bff^2$ &(D) $n=300$, $g(\bff)=\bff^2$\\

\end{tabular}

\end{figure}

\begin{example}\label{exmp::smu-factor}[A general graphical model with external factors]
We consider a general conditional dependency graph with the contribution of external factors by combining the ingredients of Examples \ref{exmp::smu} and \ref{exmp::factor-graph}. In particular, we generate $\bu=(\bx_1^T, \bx_2^T)$ with $\bx_1$ and $\bx_2$ generated from Example \ref{exmp::smu} and $\bff \sim \mathcal{N}({\bf 0}, \bI_{300})$, then set $\bx = \bu + \bQ g(\bff)$ where $\bQ$ is the same as Example \ref{exmp::factor-graph}. We also consider $g(\bff)=\bff$ and $g(\bff) =\bff^2$.
\end{example}
In this example, we would like to investigate the performance of a two-step projection method. In particular, we first project $\bx$ onto the space spanned by $\bff$ and denote  the residual by $\hat\bu$. Then we explore the conditional dependency structure of $\hat\bu^{(i)}$ and $\hat\bu^{(j)}$ given $\hat\bu^{(-i, -j)}$ by projecting them onto the space orthogonal to the space (linearly or additively) spanned  by $\hat\bu^{(-i, -j)}$. Here, we compare the performances of  methods using the external factor and those that ignore them. The average ROC curves are rendered in Figure \ref{fig:: smu-factor}.

From Figure \ref{fig:: smu-factor}, we  see that first of all, when $g(\bff)=\bff$,  the methods using  external factors outperform their counterparts without using the information with the best method being lasso.dcov.  Second, when we have nonlinear factors, using the factors do not necessarily help when we only consider  linear projection. For example,  the performances of lasso.pearson and lasso.pearson.f in panel (c) illustrates this point. On the other hand, by using sparse additive model based projection, we have a substantial gain over all the remaining methods especially for $n=300$.

\begin{figure}[tb!]
\caption{ROC curves for a general graphical model with external factors (AUCs in legends). \label{fig:: smu-factor}}
\begin{tabular}{cc}
	\includegraphics[scale=0.4]{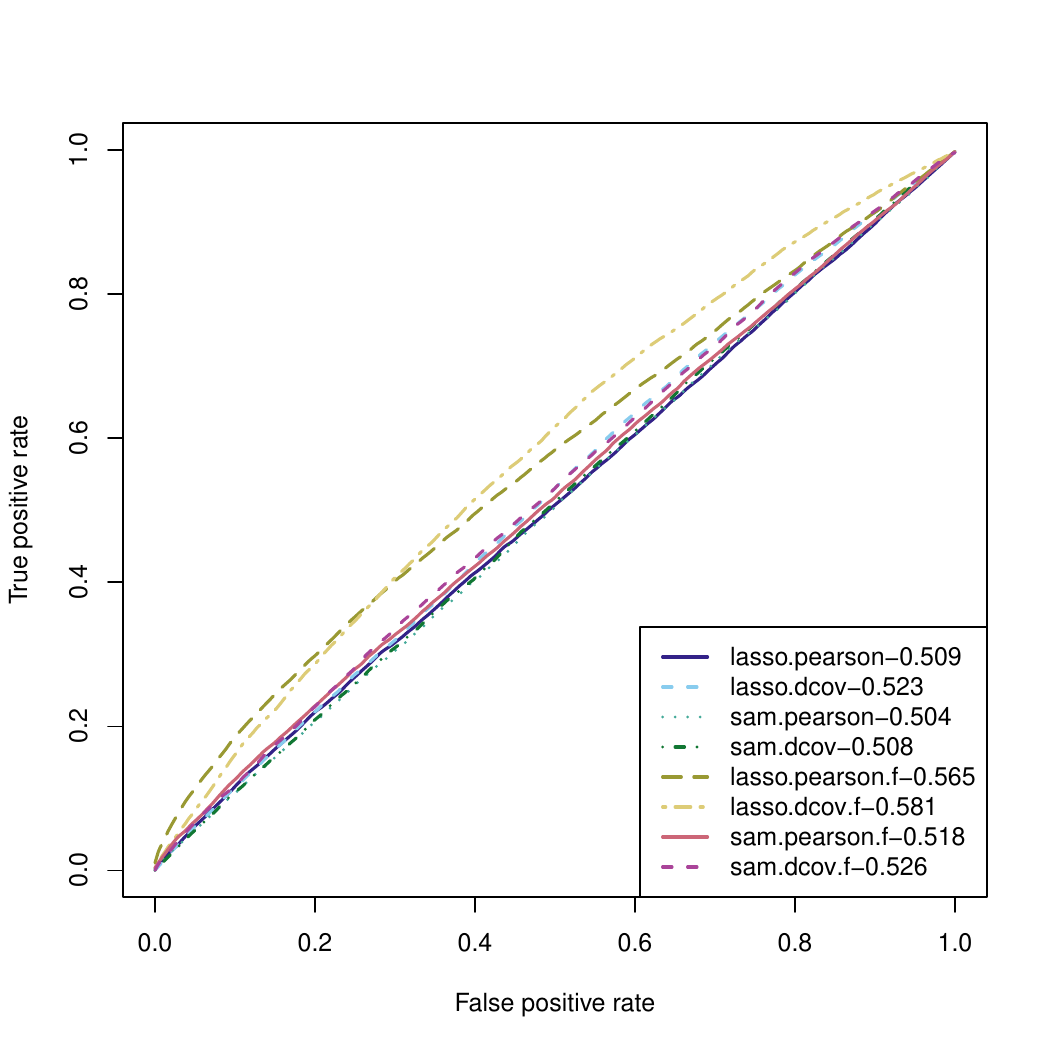}&\includegraphics[scale=0.4]{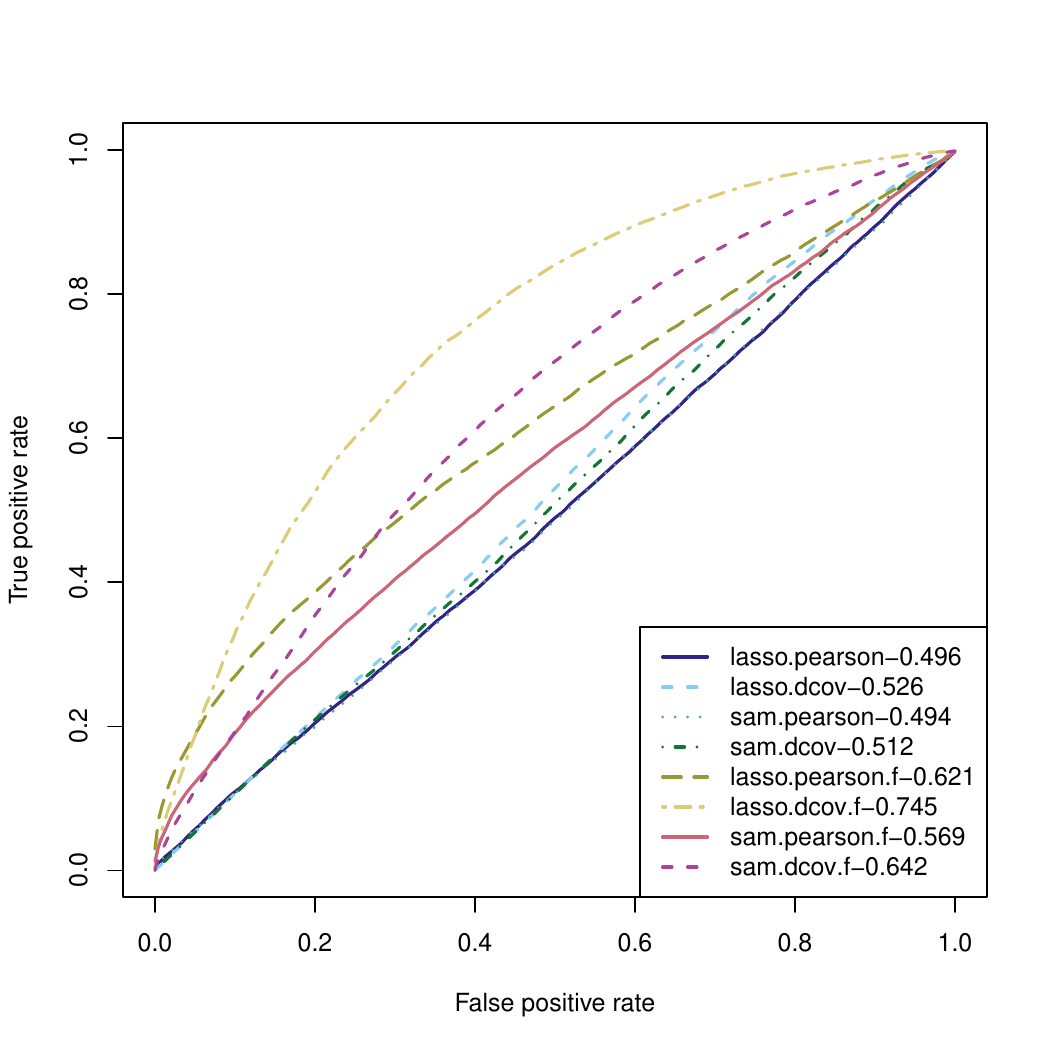}\\
	(A) $n=100$, $g(\bff)=\bff$ &(B) $n=300$, $g(\bff)=\bff$\\
\includegraphics[scale=0.4]{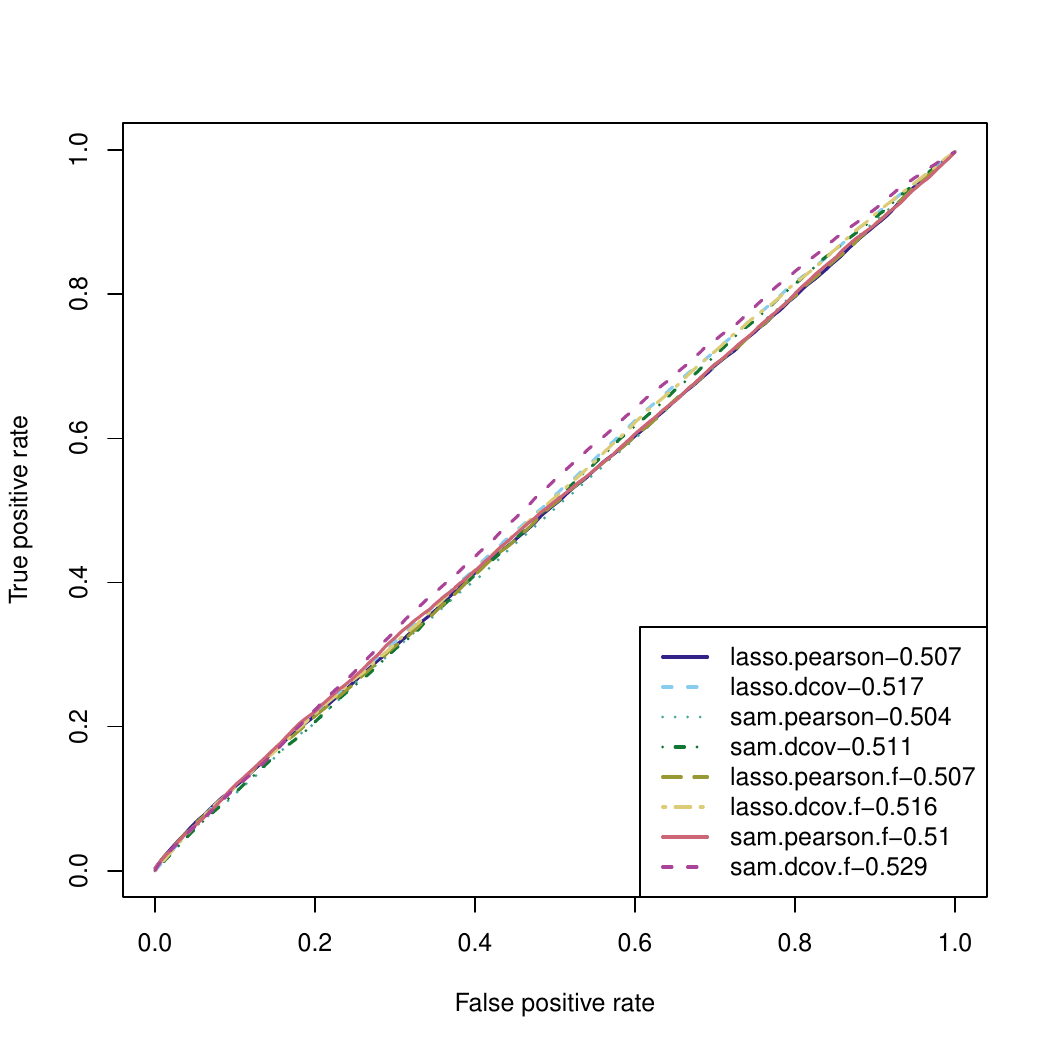}&\includegraphics[scale=0.4]{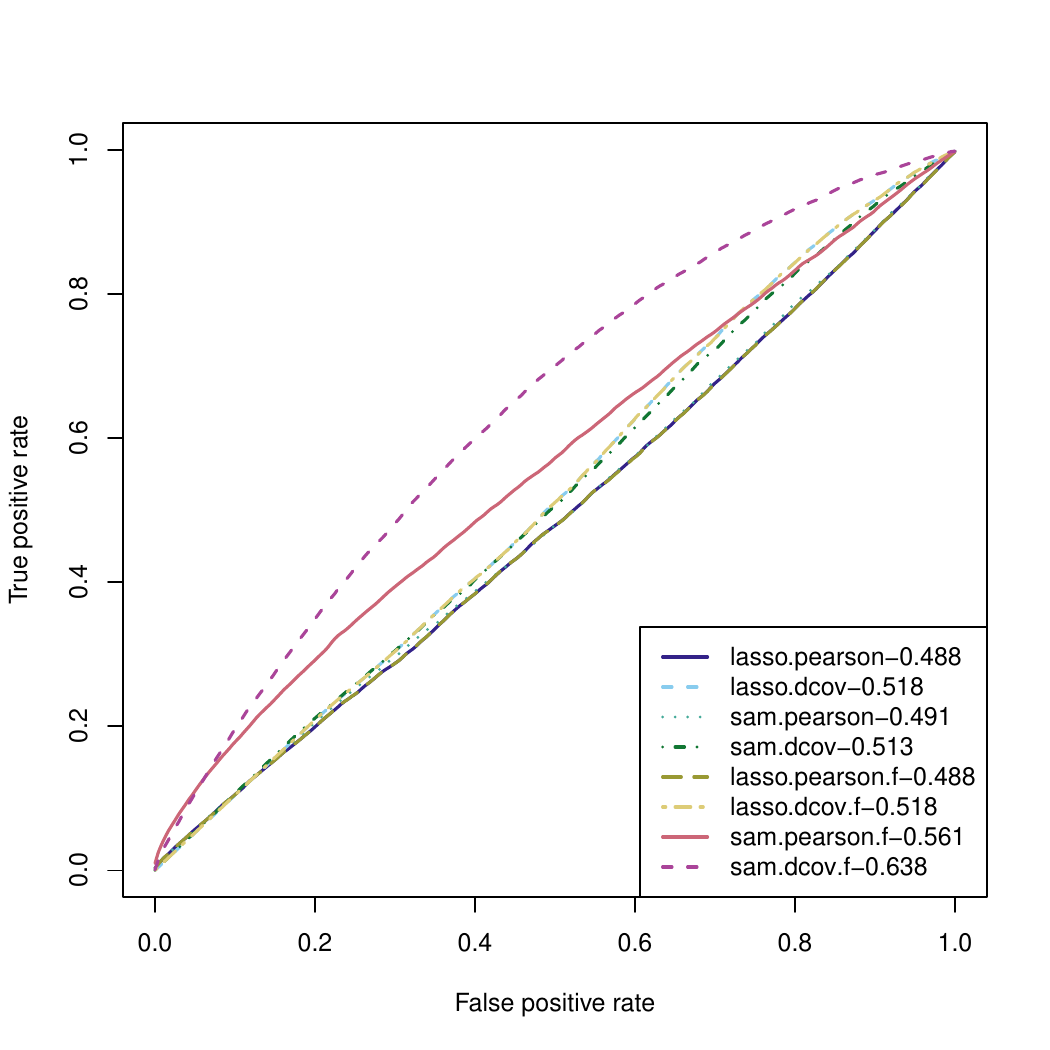}\\
	(C) $n=100$, $g(\bff)=\bff^2$ &(D) $n=300$, $g(\bff)=\bff^2$\\

\end{tabular}

\end{figure}

\section{Real Data Analysis}\label{sec::Realdata}

We collect daily excess returns of 90 stocks among the S\&P 100 index, which are available between August 19, 2004 and August 19, 2005. We chose the starting date as Google's Initial Public Offering date, and consider one year of daily excess returns since then. In particular, we consider the following Fama-French three-factor model \citep{fama1993common}
\[
r_{it}-r_{ft}=\beta_{i,\text{MKT}}(\text{MKT}_t-r_{ft})+\beta_{i,\text{SMB}}\text{SMB}_{t}+\beta_{i,\text{HML}}\text{HML}_t+u_{it},
\]
for $i=1,\dots,90$ and $t=1,\dots,252$. At time $t$, $r_{it}$ represents the return for stock $i$, $r_{ft}$ is the risk-free return rate, and MKT$_t$, SMB$_t$ and HML$_t$ constitute market, size and value factors, respectively.
\subsection{Individual stocks}\label{ind_stock}
In the first experiment, we perform P-DCov test with FDR control on all pairs of stocks and study the dependence between stocks conditional on the  Fama-French three-factors. Under significance level $\alpha=0.01$, we found out that 15.46\% of the pairs of stocks are conditionally dependent given the three factors, which implies that the three factors may not be sufficient to explain the dependencies among stocks. As a comparison, we also implemented the conditional independence test with the distance covariance based test replaced by Pearson correlation based  test. It turns out the 9.34\% of the pairs are significant under the same significance level. This shows the P-DCov test is more powerful than the Pearson correlation test in discovering significant pairs that are conditionally dependent.

We then investigate the top 5 pairs of stocks that correspond to the largest test statistic values using the P-DCov test. They are (BHI, SLB), (CVX, XOM), (HAL,  SLB), (COP, CVX), and (BHI, HAL). Interestingly, all six stocks involved are closely related to the oil industry. This reveals the high level of dependence among oil industry stocks that cannot be well explained by the Fama-French three-factor model. In addition, we examine the stock pairs that are conditionally dependent under the P-DCov test but not under the Pearson correlation test.  The two most significant pairs are (C, USB) and (MRK, PFE). The first pair is in the financial industry (Citigroup and U.S. Bancorp) and the second pair is pharmaceutical companies (Merck \& Co. and Pfizer). This shows that by using the proposed P-DCov, some interesting conditional dependency structures could be recovered.  This is consistent with the findings that the within-sector correlations are still present even after adjusting for Fama-French factors and 10 industrial factors \citep{fan2016incorporating}.

\subsection{Stock groups by industry}

One advantage of our proposed procedure is that P-DCov can investigate dependence between two multivariate vectors, not necessarily of the same dimension, conditional on external factors. As an illustration, beyond studying the relationship of stocks within industrial sectors as in Section \ref{ind_stock}, we explore dependency structures between industrial sectors conditional on the Fama-French three-factors. In particular, we group the stocks in S\&P 100 into 32 industrial groups based on the ``Sectoring by industry groups" information provided on \url{https://www.nasdaq.com}. Each of the industrial group now contains a few stocks,  with a full list provided in Table \ref{tb::groups-stocks} in Appendix.
\begin{table}[ht]
\caption{Pairs of stock groups with the smallest P-value 1e-6.\label{tb::groups-smallest-pvalue}}
\begin{center}
\begin{tabular}{lr}
\hline
Conglomerates&Aerospace\\
Large Cap Pharma&Medical Products\\
Soap and Cleaning Products&Large Cap Pharma\\
Conglomerates&Transportation\\
Banks&Finance\\
Banks&Medical Products\\
Conglomerates&Banks\\
Banks&Utility\\
Soap and Cleaning Products&Banks\\
Wireless National&Banks\\
\hline
\end{tabular}
\end{center}
\end{table}
We perform P-DCov test on all pairs of industrial groups conditional on the same Fama-French three-factors in Section \ref{ind_stock}. Table \ref{tb::groups-smallest-pvalue} presents the pairs of industrial groups (containing more than 2 stocks) which attain the smallest P-value of 1e-6, and for readers' convenience, we list the stocks corresponding to each selected groups in  Table \ref{tb::groups-stocks-smallest-pvalue}. A few interesting findings are the following. Industry `Conglomerates' (containing stocks of General Electric, Honeywell, 3M and United Technologies Corporation), is conditionally dependent of both `Aerospace' (containing stocks of Boeing, General Dynamics and Raytheon) and `Transportation' (containing stocks of FedEx, Norfolk Southern and UPS-United Parcel Service). A plausible explanation is that the companies in sector `Conglomerates' may produce supplies such as components/gadgets for sector `Aerospace' and `Transportation' and therefore the returns of these industrial sectors might be dependent. Similarly, `Large Cap Pharma' (containing stocks of Bristol-Myers Squibb, Johnson \& Johnson, Merck \& Co and Pfizer) is conditionally dependent of `Medical Products' (containing stocks of Abbott Laboratories, Baxter International and Medtronic) and `Soap and Cleaning Products' (containing stocks of Colgate-Palmolive and Procter \& Gamble). The first relationship can be explained as Pharmaceutical versus Health care and the second is due to the fact that companies in `Soap and Cleaning Products' are big suppliers of the Pharmaceutical companies in terms of their commonly used commodities. Lastly, based on the industrial division provided by Nasdaq, sector `Finance' contains mainly investment banks while sector `Banks' contains the usual regional and commercial banks. It is reasonable to believe these two sectors are closely dependent. The rest of the pairs are detected as significant although we cannot provide an obvious explanation. Nevertheless, since the Fama-French three-factors are conditioned out, the discovered conditional dependencies can be subtle. We will leave them to experts for further investigation.

\begin{table}[ht]
\caption{The stocks corresponding to each selected  industry.\label{tb::groups-stocks-smallest-pvalue}}
\begin{center}
\begin{tabular}{lr}
\hline
Banks&BAC C JPM RF USB WFC\\
Large Cap Pharma&BMY JNJ MRK PFE\\
Soap and Cleaning Products&CL PG\\
Conglomerates&GE HON MMM UTX\\
Wireless National&S T VZ\\
Medical Products&ABT BAX MDT\\
Utility&AEP AES ETR EXC SO\\
Finance&AXP COF GS MS\\
Aerospace&BA GD RTN\\
Transportation&FDX NSC UPS\\
\hline
\end{tabular}
\end{center}
\end{table}

After looking at the interesting pairs corresponding to the smallest P-values, we apply FDR control with $\alpha=0.01$ and selected 27 important pairs with results presented in Tables \ref{tb::groups-fdr} and  \ref{tb::groups-stocks}. Similar messages can be discovered and we leave out the detailed discussions due to the large number of pairs.

\section{Discussion}\label{sec::Disc}
In this work, we proposed a general framework for testing conditional independence via projection and showed a new way to create dependency graphs. The current theoretical results assume that contribution of factors is sparse linear. How to extend the theory to the case of sparse additive model projection would be an interesting future work. 
%Another interesting direction is to use the proposed test to create dependency graphs among groups of nodes, which could have applications in genetics.
Another interesting future work is to extend the methodology and theory to the case where the dimensions of $\bx$ and $\by$ grow with $n$.  

An R package \pkg{pgraph} for implementing the proposed methodology is available on CRAN.
%\appendix
\section*{Appendix}\label{sec::Proofs}

\begin{lemma}\label{lemma:expectation-U}
Under Condition \ref{Assump:dn}, we have $\max_{i,j}\|(\bB_x-\hbB_x)(\bff_i-\bff_j)\| = O_p(a_n)$ and $\mathbb{E}\max_{i,j}\|(\bB_x-\hbB_x)(\bff_i-\bff_j)\| = O(a_n)$.
\end{lemma}
\begin{proof}
From Condition \ref{Assump:dn}, it is obvious that $\max_{i,j}\|(\bB_x-\hbB_x)(\bff_i-\bff_j)\|=O_p(a_n)$. Let $U_n=\max_{i,j}\|(\bB_x-\hbB_x)(\bff_i-\bff_j)\|$ and 	$\tilde{U}_n= U_n/{a_n}$. Then, we have
\begin{align*}
\mathbb{E}(\tilde{U}_n)&=\int_{0}^{\infty} \mathbb{P}(\tilde{U}_n>t) dt\cr
&=\int_{0}^{1} \mathbb{P}(\tilde{U}_n>t) dt +\int_{1}^{\infty}\mathbb{P}(\tilde{U}_n>t) dt\cr
&\leq 1+\int_{1}^{\infty} C_1^{-t} dt <\infty.
\end{align*}
As a result, the lemma is proved.
\end{proof}

For the remaining proofs, we apply Taylor expansion to $\|\hbeps_{i,x}-\hbeps_{j,x}\|$ at $\beps_{i,x}-\beps_{j,x}$ and get
\begin{align}
\|\hbeps_{i,x}-\hbeps_{j,x}\|&=\|\beps_{i,x}-\beps_{j,x}\|+\frac{\bc^{\top}_{i,j,x}}{\|\bc_{i,j,x}\|} (\bB_x-\hbB_x)(\bff_i-\bff_j)=\|\beps_{i,x}-\beps_{j,x}\|+D_{i,j,x},\cr
\|\hbeps_{i,y}-\hbeps_{j,y}\|&=\|\beps_{i,y}-\beps_{j,y}\|+\frac{\bc^{\top}_{i,j,y}}{\|\bc_{i,j,y}\|} (\bB_y-\hbB_y)(\bff_i-\bff_j)=\|\beps_{i,y}-\beps_{j,y}\|+D_{i,j,y}, \label{eq::taylor-expan-eps}
\end{align}
where $\bc_{i,j,x}=\lambda_{i,j,x}(\hbeps_{i,x}-\hbeps_{j,x})+(1-\lambda_{i,j,x})(\beps_{i,x}-\beps_{j,x})$ and $\bc_{i,j,y}=\lambda_{i,j,y}(\hbeps_{i,y}-\hbeps_{j,y})+(1-\lambda_{i,j,y})(\beps_{i,y}-\beps_{j,y})$, for $\lambda_{i,j,x}\in[0,1]$ and $\lambda_{i,j,y}\in[0,1]$.

\begin{proof}[of Theorem \ref{thm::consistency}] 
Using the Taylor expansion in \eqref{eq::taylor-expan-eps}, we have the following decomposition
\begin{align*}
\mathcal{V}^2_n(\hbeps_{x},\hbeps_{y})-\mathcal{V}^2_n(\beps_{x},\beps_{y})=T_1+T_2+T_3,
\end{align*}
where
{\footnotesize
\begin{align}
T_1&=\frac{1}{n^2}\sum_{i,j=1}^n D_{i,j,x} \|\beps_{i,y}-\beps_{j,y}\|+\frac{1}{n^2}\sum_{i,j=1}^n D_{i,j,x} \frac{1}{n^2}\sum_{k,l=1}^n \|\beps_{k,y}-\beps_{l,y}\|-\frac{2}{n^3}\sum_{i=1}^n\sum_{j,k=1}^n D_{i,j,x}\|\beps_{i,y}-\beps_{k,y}\|,\label{eq::T1}\\
T_2&=\frac{1}{n^2}\sum_{i,j=1}^n D_{i,j,y} \|\beps_{i,x}-\beps_{j,x}\|+\frac{1}{n^2}\sum_{i,j=1}^n D_{i,j,y}\frac{1}{n^2}\sum_{k,l=1}^n \|\beps_{k,x}-\beps_{l,x}\|-\frac{2}{n^3}\sum_{i=1}^n\sum_{j,k=1}^n D_{i,j,y}\|\beps_{i,x}-\beps_{k,x}\|,\label{eq::T2}\\
T_3&=\frac{1}{n^2}\sum_{i,j=1}^n D_{i,j,x} D_{i,j,y}+\frac{1}{n^2}\sum_{i,j=1}^n D_{i,j,x} \frac{1}{n^2}\sum_{i,j=1}^n D_{i,j,y}-\frac{2}{n^3}\sum_{i=1}^n \sum_{j,k=1}^n D_{i,j,x} D_{i,k,y}.\label{eq::T3}
\end{align}
}
By  Condition \ref{Assump:dn},  we have $\max_{i,j} |D_{i,j,x}| \le 2\|(\bB_x-\hat\bB_x)\bF\|_{2,\infty}\le  O_p(a_n)$. Therefore,
\begin{align*}
|T_1| &= O_p(a_n)\left(\frac{4}{n^2} \sum_{i,j =1}^n \|\beps_{i,y}-\beps_{j,y}\|\right),\cr
|T_2| &= O_p(a_n)\left(\frac{4}{n^2} \sum_{i,j =1}^n \|\beps_{i,x}-\beps_{j,x}\|\right),\cr
|T_3| &=  O_p(a^2_n).
\end{align*}

Another fact we easily observe is that: $n^{-2}\sum_{i,j=1}^n \|\beps_{i,x}-\beps_{j,x}\|=O_p(1)$, since $\mathbb{E}\|\beps_{i,x}-\beps_{j,x}\|$ is uniformly bounded over all $(i,j)$ pairs and so is $\mathbb{E}(n^{-2} \sum_{i,j=1}^n \|\beps_{i,x}-\beps_{j,x}\|)$.

As a result, we know $\mathcal{V}^2_n(\hbeps_{x},\hbeps_{y})-\mathcal{V}^2_n(\beps_{x},\beps_{y})\convinprob 0$. This combined with Lemma \ref{thm::old-consistency} leads to
\[
\mathcal{V}^2_n(\hbeps_{x},\hbeps_{y})\convinprob \mathcal{V}^2(\beps_{x},\beps_{y}).
\]
\end{proof}

\emph{Remark:} The result of Theorem \ref{thm::consistency}  cannot be implied from that of Theorem \ref{thm::null-distribution}, since independence between $\beps_{x}$ and $\beps_{y}$ is not assumed.

\begin{lemma}\label{Lemma::DSepa}
 For the $\bc_{i,j,x}$ and $\bc_{i,j,y}$ defined in \eqref{eq::taylor-expan-eps}, we have the following approximation error bound on the normalized version.
\begin{align} \left\|\frac{\bc_{i,j,x}}{\|\bc_{i,j,x}\|}-\frac{\beps_{i,x}-\beps_{j,x}}{\|\beps_{i,x}-\beps_{j,x}\|}\right\|&\leq\frac{2}{\|\beps_{i,x}-\beps_{j,x}\|} \max_{i,j}\|(\bB_x-\hbB_x)(\bff_i-\bff_j)\|,\label{eq::cijx-error} \\
\left\|\frac{\bc_{i,j,y}}{\|\bc_{i,j,y}\|}-\frac{\beps_{i,y}-\beps_{j,y}}{\|\beps_{i,y}-\beps_{j,y}\|}\right\|&\leq\frac{2}{\|\beps_{i,y}-\beps_{j,y}\|} \max_{i,j}\|(\bB_y-\hbB_y)(\bff_i-\bff_j)\|.\label{eq::cijy-error}
\end{align}
\end{lemma}

\begin{proof}
It suffices to show \eqref{eq::cijx-error}. First, we will show\begin{equation}\label{Lemma::Dsepa:Step1}
\left\|\frac{\bc_{i,j,x}}{\|\bc_{i,j,x}\|}-\frac{\beps_{i,x}-\beps_{j,x}}{\|\beps_{i,x}-\beps_{j,x}\|}\right\|\leq\left\|\frac{\hbeps_{i,x}-\hbeps_{j,x}}{\|\hbeps_{i,x}-\hbeps_{j,x}\|}-\frac{\beps_{i,x}-\beps_{j,x}}{\|\beps_{i,x}-\beps_{j,x}\|}\right\|.
\end{equation}
Denote by $\alpha_1$ and $\alpha_2$ the angle between $\bc_{i,j,x}$ and $\beps_{i,x}-\beps_{j,x}$, and the angle between $\hbeps_{i,x}-\hbeps_{j,x}$ and $\beps_{i,x}-\beps_{j,x}$, respectively. It is easy to see that $0\leq\alpha_1\leq \alpha_2\leq \pi$, and hence $\cos\alpha_1\geq\cos\alpha_2$. By cosine formula,
{\footnotesize
\begin{align*}
\left\|\frac{\bc_{i,j,x}}{\|\bc_{i,j,x}\|}-\frac{\beps_{i,x}-\beps_{j,x}}{\|\beps_{i,x}-\beps_{j,x}\|}\right\|^2=2-2\cos \alpha_1,\mbox{ and }
\left\|\frac{\hbeps_{i,x}-\hbeps_{j,x}}{\|\hbeps_{i,x}-\hbeps_{j,x}\|}-\frac{\beps_{i,x}-\beps_{j,x}}{\|\beps_{i,x}-\beps_{j,x}\|}\right\|^2=2-2\cos\alpha_2.
\end{align*}
}
Therefore, (\ref{Lemma::Dsepa:Step1}) is proved and it remains to show that
\begin{equation}\label{Lemma::Dsepa:Step2}
\left\|\frac{\hbeps_{i,x}-\hbeps_{j,x}}{\|\hbeps_{i,x}-\hbeps_{j,x}\|}-\frac{\beps_{i,x}-\beps_{j,x}}{\|\beps_{i,x}-\beps_{j,x}\|}\right\|\leq \frac{2}{\|\beps_{i,x}-\beps_{j,x}\|} \max_{i,j \in\{1,\dots,n\}}\|(\bB_x-\hbB_x)(\bff_i-\bff_j)\|.
\end{equation}
Left hand side of (\ref{Lemma::Dsepa:Step2}) can be rewritten as
\begin{align*}
&\left\|\frac{\hbeps_{i,x}-\hbeps_{j,x}}{\|\hbeps_{i,x}-\hbeps_{j,x}\|}-\frac{\beps_{i,x}-\beps_{j,x}}{\|\beps_{i,x}-\beps_{j,x}\|}\right\|\cr
=&\left\|\frac{[(\hbeps_{i,x}-\hbeps_{j,x})-(\beps_{i,x}-\beps_{j,x})]\|\hbeps_{i,x}-\hbeps_{j,x}\|-(\|\hbeps_{i,x}-\hbeps_{j,x}\|-\|\beps_{i,x}-\beps_{j,x}\|)(\hbeps_{i,x}-\hbeps_{j,x})}{\|\hbeps_{i,x}-\hbeps_{j,x}\|\|\beps_{i,x}-\beps_{j,x}\|}\right\|\cr
\leq&\frac{1}{\|\beps_{i,x}-\beps_{j,x}\|}(\|(\hbeps_{i,x}-\hbeps_{j,x})-(\beps_{i,x}-\beps_{j,x})\|+|\|\hbeps_{i,x}-\hbeps_{j,x}\|-\|\beps_{i,x}-\beps_{j,x}\||)\cr
\leq&\frac{2}{\|\beps_{i,x}-\beps_{j,x}\|} \max_{i,j \in\{1,\dots,n\}}\|(\bB_x-\hbB_x)(\bff_i-\bff_j)\|.
\end{align*}
Combining (\ref{Lemma::Dsepa:Step1}) and (\ref{Lemma::Dsepa:Step2}), the lemma is proved.
\end{proof}

\begin{lemma}\label{Lemma::SumInv1}
 Under Conditions \ref{Assump:Moments} and \ref{Assump:Tail Condition}, and  the null hypothesis that $\beps_x \independent \beps_y$, for any $\gamma>0$,
\begin{align*}
\frac{1}{ n^\gamma\log n}\left[\frac{1}{n^2}\sum_{i,j=1}^n \frac{1}{\|\beps_{i,x}-\beps_{j,x}\|}\right]\overset{P}{\rightarrow}0,\quad \frac{1}{ n^\gamma\log n}\left[\frac{1}{n^2}\sum_{i,j=1}^n \frac{1}{\|\beps_{i,y}-\beps_{j,y}\|}\right]\overset{P}{\rightarrow}0.
\end{align*}
\end{lemma}
\begin{proof}
We will only show the first result involving $\beps_x$ with the other one follows similarly. For any $\delta >0$, let
\[
R_n=\frac{1}{n^2}\sum_{i,j=1}^n \frac{1}{\|\beps_{i,x}-\beps_{j,x}\|},\ \ \bar{R}_n=\frac{1}{n^2}\sum_{i,j=1}^n \left[\frac{1}{\|\beps_{i,x}-\beps_{j,x}\|}\wedge n^{2+\delta}\right].
\]
Then for $\forall$ $\epsilon>0$,
\begin{equation}\label{Lemma::SumInv1:Step11}
\mathbb{P}[|R_n-\bar{R}_n|>\epsilon]\leq n^2\mathbb{P}[\|\beps_{i,x}-\beps_{j,x}\|<n^{-2-\delta}]\leq C n^2 n^{-2-\delta}=C n^{-\delta},
\end{equation}
due to the Condition \ref{Assump:Tail Condition} that the density function of $\|\beps_{i,x}-\beps_{j,x}\|$ is pointwise bounded. Therefore, $|R_n-\bar{R}_n|\overset{P}{ \rightarrow}0$, which leads to
\begin{equation}\label{Lemma::SumInv1:Step12}
\left|\frac{R_n}{n^\gamma \log n}-\frac{\bar{R}_n}{n^\gamma \log n}\right|\overset{P}{ \rightarrow}0.
\end{equation}
On the other hand,
\begin{align}\label{Lemma::SumInv1:Step2}
&~~~\mathbb{E}[\frac{1}{\log n}\frac{1}{\|\beps_{i,x}-\beps_{j,x}\|}\wedge n^{2+\delta}]\cr&= \frac{1}{\log n}\mathbb{P}(\frac{1}{\|\beps_{i,x}-\beps_{j,x}\|}>n^{2+\delta})n^{2+\delta}+\frac{1}{\log n}\int^{\infty}_{n^{-2-\delta}} \frac{1}{t} h_{\|\beps_{i,x}-\beps_{j,x}\|}(t)dt\cr
&\leq \frac{C}{\log n}+\frac{1}{\log n}\int^{C_0}_{n^{-2-\delta}} \frac{1}{x} h_{\|\beps_{i,x}-\beps_{j,x}\|}(x)dx+\frac{1}{\log n}\int^{\infty}_{C_0} \frac{1}{t} h_{\|\beps_{i,x}-\beps_{j,x}\|}(t)dt\cr
&\leq \frac{C}{\log n}+\frac{C}{\log n}\int^{C_0}_{n^{-2-\delta}}\frac{1}{x}dx+\frac{1}{C_0 \log n }\mathbb{P}(\|\beps_{i,x}-\beps_{j,x}\|>C_0)\cr
&\leq \frac{C}{\log n}+\frac{C}{ \log n}[\log(C_0)+\log(n^{2+\delta})]+\frac{1}{C_0 \log n}\cr
&\leq \frac{C}{\log n}+C'+\frac{1}{C_0 \log n},
\end{align}
where $h_{\|\beps_{i,x}-\beps_{j,x}\|}$ is the density of $\|\beps_{i,x}-\beps_{j,x}\|$. In the above derivation, the first inequality can be easily seen from (\ref{Lemma::SumInv1:Step11}) and the second inequality utilizes Condition \ref{Assump:Tail Condition}.

Therefore, $\bar{R}_n/{\log n}$ is bounded in $L_1$ and since $n^\gamma\rightarrow \infty$, $\bar{R}_n/[n^\gamma \log(n)]$ converges to 0 in $L_1$ and hence in probability, i.e.,
\begin{align}
\frac{\bar{R}_n}{n^\gamma \log(n)}\overset{P}{\rightarrow}0.
\end{align}
This, combined with (\ref{Lemma::SumInv1:Step12}) yields
\begin{align}
\frac{R_n}{n^\gamma\log(n)}\overset{P}{\rightarrow}0.
\end{align} This completes the proof of Lemma \ref{Lemma::SumInv1}.
\end{proof}

To prove Theorem \ref{thm::null-distribution}, we first introduce two propositions.

\begin{proposition}\label{prop::T_1_T_2}
Under Conditions \ref{Assump:Moments} and \ref{Assump:Tail Condition}, and  the null hypothesis that $\beps_x \independent \beps_y$,
\[
T_1=O_p(a_n/n),\ \ \ T_2=O_p(a_n/n)
\]
\end{proposition}
\begin{proof}
From \eqref{eq::T1}, we rewrite $T_1$ as
{\footnotesize
\begin{align*}
T_1&=\frac{1}{n^2}\sum_{i,j=1}^n D_{i,j,x}\left(\|\beps_{i,y}-\beps_{j,y}\|+\frac{1}{n^2}\sum_{k,l=1}^n \|\beps_{k,y}-\beps_{l,y}\| -\frac{1}{n}\sum_{k=1}^n \|\beps_{i,y}-\beps_{k,y}\|-\frac{1}{n}\sum_{k=1}^n \|\beps_{j,y}-\beps_{k,y}\|\right)\cr
&\doteq \frac{1}{n^2}\sum_{i,j=1}^n \ D_{i,j,x} A_{i,j,y},
\end{align*}
}
with $A_{i,j,y}$ self-defined by the equation.

Let us consider term
\begin{equation}
\mathbb{E}(T^2_1)=\frac{1}{n^4} \sum_{i\neq j, k\neq l} \mathbb{E}(D_{i,j,x} D_{k,l,x} A_{i,j,y} A_{k,l,y})=\frac{1}{n^4} \sum_{i\neq j, k\neq l} \mathbb{E}(D_{i,j,x} D_{k,l,x}) \mathbb{E}(A_{i,j,y} A_{k,l,y}).
\end{equation}

We can separate the above quantity into three parts. It is easy to see that $D_{i,j,x}$ are identically distributed with respect to different pairs of $(i,j)$ when $i\neq j$.  Let us define the following three sets of index quadruples:
\begin{itemize}
\item $I_1 = \{(i,j,k,l)| \mbox{there are four distinct values in } \{i,j,k,l\}\}$.
\item $I_2 = \{(i,j,k,l)| i\neq j, k\neq l, \mbox{and there are three distinct values in } \{i,j,k,l\}\}$.
\item $I_3 = \{(i,j,k,l)| i\neq j, k\neq l, \mbox{and there are two distinct values in } \{i,j,k,l\}\}$.
\end{itemize}

Let us suppose $\mathbb{E}(D_{i,j,x}D_{k,l,x})=c_1$, for $(i,j,k,l)\in I_1$; $\mathbb{E}(D_{i,j,x}D_{k,l,x})=c_2$, for $(i,j,k,l)\in I_2$. $\mathbb{E}(D_{i,j,x}D_{k,l,x})=c_3$, for $(i,j,k,l)\in I_3$. By Condition \ref{Assump:dn}, we know $c_1$, $c_2$ and $c_3$ are all of order $O(a^2_n)$. Also, $\mathbb{E}(A_{i,j,y})=O(1)$. Then we have
\begin{equation}
\mathbb{E}(T^2_1)=\mathbb{E}\left(\frac{c_1}{n^4}\sum_{I_1}A_{i,j,y}A_{k,l,y}+\frac{c_2}{n^4}\sum_{I_2}A_{i,j,y}A_{k,l,y}+\frac{c_3}{n^4}\sum_{I_3} A_{i,j,y}A_{k,l,y}\right).
\end{equation}

On the other hand, we observe that $\sum_{j=1}^n A_{i,j,y}=0$ by definition and $A_{i,j,y}=A_{j,i,y}$, so we have
\[
\sum_{I_2} A_{i,j,y}A_{k,l,y}=\sum_{i=1}^n (\sum_{j=1}^n A_{i,j,y})^2-\sum_{i=1}^n\sum_{j=1}^n A^2_{i,j,y}=-\sum_{i=1}^n\sum_{j=1}^n A^2_{i,j,y}.
\]
By Condition \ref{Assump:Moments}, we know all the second order terms of distances of differences ($\|\beps_{i,y}-\beps_{j,y}\|^2$, $\|\beps_{i,y}-\beps_{j,y}\|\cdot \|\beps_{i,y}-\beps_{k,y}\|$ as examples) have bounded expectation, and thus all the second order terms of $A_{i,j,y}$'s also have bounded expectations. Therefore, $\mathbb{E}( n^{-4} \sum_{I_3}  A_{i,j,y}A_{k,l,y})=O(n^{-2})$.
Finally, since $\sum_{i=1}^n \sum_{j=1}^n A_{i,j,y}=0$,
\begin{align*}
\sum_{I_1}A_{i,j,y}  A_{k,l,y}=&(\sum_{i=1}^n \sum_{j=1}^n A_{i,j,y})^2-\sum_{I_2} A_{i,j,y}A_{k,l,y}-\sum_{I_3} A_{i,j,y}A_{k,l,y}-\sum_{i=1}^n A^2_{i,i,y}\cr
=&-\sum_{I_2} A_{i,j,y}A_{k,l,y}-\sum_{I_3} A_{i,j,y}A_{k,l,y}-\sum_{i=1}^n A^2_{i,i,y}.
\end{align*}
This combined with our previous calculations leads to $\mathbb{E}( n^{-4} \sum_{I_1} A_{i,j,y} A_{k,l,y})=O(n^{-2})$. As a result, we have $\mathbb{E}(T^2_1)=O(a^2_n/n^2)$. Together with Chebychev's inequality, we know $T^2_1=O_p(a^2_n/n^2)$ and equivalently, $T_1=O_p(a_n/n)$. Similarly, we could show that $T_2=O_p(a_n/n)$.
\end{proof}

\begin{proposition}\label{prop::T3}
Under Conditions \ref{Assump:Moments}, \ref{Assump:Tail Condition}, \ref{Assump:dn}, and \ref{Assump:l1}, and  the null hypothesis that $\beps_x \independent \beps_y$,
\[
T_3 = O_p\{(n^{-1/2}{a^2_n})\vee (a^3_n (\log n) n^{\gamma}) \vee (n^{-1/2}{a_n e_n \log K})\}.
\]
\end{proposition}
\begin{proof}
Recall that
\begin{align*}
T_3&=\frac{1}{n^2}\sum_{i,j=1}^n D_{i,j,x} D_{i,j,y}+\frac{1}{n^2}\sum_{i,j=1}^n D_{i,j,x} \frac{1}{n^2}\sum_{i,j=1}^n D_{i,j,y}-\frac{2}{n^3}\sum_{i=1}^n \sum_{j,k=1}^n D_{i,j,x} D_{i,k,y}\cr
& \doteq \frac{1}{n^2}\sum_{i,j=1}^n D_{i,j,x} B_{i,j,y},
\end{align*}
with $B_{i,j,y}$ self-defined in the above equation. We can easily see that $\sum_{i=1}^n B_{i,j,y}=0$, for any $j$. Let $B_{\max}=\max_{i,j} |B_{i,j,y}|$, then we define $\tilde{B}_{i,j,y} = B_{i,j,y}/{(2B_{\max})}+0.5$. In this way, we know $\tilde{B}_{i,j,y}\in[0,1]$ and $\sum_{i=1}^n \tilde{B}_{i,j,y} =1/2$ for any $j$. By Condition \ref{Assump:dn}, we know that $B_{\max}=O_p(a_n)$.  Also, since all $\tilde{B}_{i,j,y}$ are non-negative, by Cauchy-Schwartz, we can upper bound $\|\widetilde{\bB}\|_F$ with the case when $\tilde{B}_{i,j,y}$ have the same values across $i$. Thus, $\|\widetilde{\bB}\|_F=O_p(\sqrt{n})$.

Then we can rewrite $T_3$ in the following form:
\begin{align*}
T_3 &= \frac{2B_{\max}}{n^2}\sum_{i,j=1}^n D_{i,j,x} \tilde{B}_{i,j,y}-\frac{B_{\max}}{n^2}\sum_{i,j=1}^n D_{i,j,x}\doteq T_{31}-T_{32}.
\end{align*}

Let us look at $T_{31}$ first. If we denote $\bD$ and $\widetilde{\bB}$ as the matrix of dimension $n\times n$ composed of elements $D_{i,j,x}$ and $\widetilde{B}_{i,j,y}$, we know that
\begin{equation}\label{T31}
|T_{31}|\leq \frac{2B_{\max}}{n^2} \|\bD\|_F \|\widetilde{\bB}\|_F = O_p({a_n}/{n^2}) O_p(a_n n) O_p(\sqrt{n})= O_p(\frac{a^2_n}{\sqrt{n}}).
\end{equation}
Then, let us proceed to term $T_{32}$. Here, we write $D_{i,j,x}$ in another form as a sum of two terms and bound them separately.
\begin{align}\label{T3Sepa}
D_{i,j,x}&=\frac{(\beps_{i,x}-\beps_{j,x})^{\top}}{\|\beps_{i,x}-\beps_{j,x}\|}(\bB_x-\hbB_x)(\bff_i-\bff_j)+\left(\frac{\bc_{i,j,x}}{\|\bc_{i,j,x}\|}-\frac{\beps_{i,x}-\beps_{j,x}}{\|\beps_{i,x}-\beps_{j,x}\|}\right)(\bB_x-\hbB_x)(\bff_i-\bff_j)\cr
&\equiv\frac{(\beps_{i,x}-\beps_{j,x})^{\top}}{\|\beps_{i,x}-\beps_{j,x}\|}(\bB_x-\hbB_x)(\bff_i-\bff_j)+r_{i,j,x}.
\end{align}
As a result, we know
\[
T_{32} = \frac{B_{\max}}{n^2}\sum_{i,j=1}^n r_{i,j,x} + \frac{B_{\max}}{n^2} \sum_{i,j=1}^n  \frac{(\beps_{i,x}-\beps_{j,x})^{\top}}{\|\beps_{i,x}-\beps_{j,x}\|}(\bB_x-\hbB_x)(\bff_i-\bff_j).
\]
By Lemma \ref{Lemma::DSepa}, we know
\begin{equation}\label{eq::dijx}
|r_{i,j,x}|\leq \max_{i, j} \|(\bB_x-\widehat{\bB}_x)(\bff_i-\bff_j)\|^2 \frac{2}{\|\beps_{i,x}-\beps_{j,x}\|},
\end{equation}
where $\max_{i, j} \|(\bB_x-\widehat{\bB}_x)(\bff_i-\bff_j)\|^2=O_p(a^2_n)$.

Together with Lemma \ref{Lemma::SumInv1}, the first term in $T_{32}$ has rate 
\[
n^{-2}B_{\max} \sum_{i,j=1}^n r_{i,j,x} = O_p(a^3_n(\log n)n^{\gamma}).
\]

The second term in $T_{32}$ can be rewritten in terms of trace:
\begin{align}\label{T322}
&~~~\left\|\frac{B_{\max}}{n^2} \sum_{i,j=1}^n  \frac{(\beps_{i,x}-\beps_{j,x})^{\top}}{\|\beps_{i,x}-\beps_{j,x}\|}(\bB_x-\hbB_x)(\bff_i-\bff_j) \right\|\\&= \left|B_{\max}\Tr{\left((\bB_x-\hbB_x)\frac{1}{n^2}\sum_{i,j=1}^n (\bff_i-\bff_j)\frac{(\beps_{i,x}-\beps_{j,x})^{\top}}{\|\beps_{i,x}-\beps_{j,x}\|}\right)}\right|,\cr
& \doteq \left| B_{\max}\Tr\left((\bB_x-\hbB_x) \bW\right)\right|,\cr
& \leq B_{\max} \sum_{l=1}^p\|B_{x,l}-\hB_{x,l}\|_1 \max_{i,j} |W(i,j)|,
\end{align}
where $\bW$ is self-defined and $W(i,j)$ is the element on the $i$-th row and $j$-column of matrix $\bW$. Let us take $(i,j)=(1,1)$ as an example, and look at $W(1,1)=\frac{1}{n^2}\sum_{i,j=1}^n (f_{i,1}-f_{j,1})\frac{\epsilon_{i,x,1}-\epsilon_{j,x,1}}{\|\beps_{i,x}-\beps_{j,x}\|}$. We easily see that $\mathbb{E} W(1,1)=0$, due to facts: $\beps_{i,x}$ and $\beps_{j,x}$ are mutually independent of $\bff$ with any observation indices;
and $\mathbb{E}[(\beps_{i,x}-\beps_{j,x})/\|\beps_{i,x}-\beps_{j,x}\|]=0$. Furthermore,

\[
\mathbb{E}(W(1,1)^2) = \frac{1}{n^4} \sum_{i,j,k,l=1}^n (f_{i,1}-f_{j,1})\frac{\epsilon_{i,x,1}-\epsilon_{j,x,1}}{\|\beps_{i,x}-\beps_{j,x}\|}(f_{k,1}-f_{l,1})\frac{\epsilon_{k,x,1}-\epsilon_{l,x,1}}{\|\beps_{k,x}-\beps_{l,x}\|}.
\]

Similar to the reasoning in Proposition \ref{prop::T_1_T_2}, we have $n^4$ terms in $I_1$. But in this scenario, $\mathbb{E} (f_{i,1}-f_{j,1})\frac{\epsilon_{i,x,1}-\epsilon_{j,x,1}}{\|\beps_{i,x}-\beps_{j,x}\|}(f_{k,1}-f_{l,1})\frac{\epsilon_{k,x,1}-\epsilon_{l,x,1}}{\|\beps_{k,x}-\beps_{l,x}\|} =0$ due to independence, therefore we know
\[
\mathbb{E}(W(1,1)^2) = O(1/n).
\]
As a result, we know $|W(1,1)|= O_p(n^{-1/2})$, and thus $\max_{i,j} |W(i,j)| = O_p( n^{-1/2} \log K)$. Furthermore, we can bound the term in \eqref{T322} with rate $O_p(n^{-1/2}{a_n e_n \log K})$.

Combining $T_{31}$ and $T_{32}$, we know $T_3 = O_p\{(n^{-1/2}{a^2_n})\vee (a^3_n (\log n) n^{\gamma}) \vee (n^{-1/2}{a_n e_n \log K})\}$.
\end{proof}

\begin{proof}[Proof of Theorem \ref{thm::null-distribution}]  Recall the notations we used in the proof of Theorem \ref{thm::consistency},
 \begin{align*}
\mathcal{V}^2_n(\hbeps_{x},\hbeps_{y})-\mathcal{V}^2_n(\beps_{x},\beps_{y})=T_1+T_2+T_3.
\end{align*}
By Propositions \ref{prop::T_1_T_2} and \ref{prop::T3}, Conditions \ref{Assump:dn} and \ref{Assump:l1}, we have for any $\gamma>0$,
\[
n(T_1+T_2+T_3)=O_p(a_n)+O_p\{(a^2_n \sqrt{n})\vee (n^{1+\gamma}(\log n)a^3_n) \vee (a_n e_n \log K \sqrt{n})\}=o_p(1).
\]
Combined with Lemma \ref{thm::old-asymptotic},  the theorem is proved.
\end{proof}

\begin{proof}
[Proof of Corollary \ref{coro::new-teststat}]
	The result follows directly from the proofs of Theorems \ref{thm::consistency} and  \ref{thm::null-distribution} and an application of Slutsky's theorem.	
\end{proof}

\begin{proof}[Proof of Theorem \ref{thm::TestSigLev}] The proof of Theorem \ref{thm::TestSigLev} follows similarly as Theorem 6 in \cite{szekely2007measuring}. Here we omit the details for brevity.		
\end{proof}

%Further material such as technical details, extended proofs, code, or additional  simulations, figures and examples may appear online, and should be briefly mentioned as Supplementary Material where appropriate.  Please submit any such content as a PDF file along with your paper, entitled `Supplementary material for Title-of-paper'.  After the acknowledgements, include a section `Supplementary material' in your paper, with the sentence `Supplementary material available at \Bka\ online includes $\ldots$', giving a brief indication of what is available.  However it should be possible to read and understand the paper without reading the supplementary material.

\section*{Acknowledgement}
The authors would like to thank the Editor, the AE and anonymous referees for their insightful comments which greatly improved the scope of this paper. This work was partially supported by National Science Foundation grants  NSF grant DMS-1308566, NSF CAREER grant DMS-1554804, DMS-1406266, NSF DMS-1662139 and National Institutes of Health grant R01-GM072611-14.%
\begin{table}[ht]
\caption{The selected important pairs of industry groups  with FDR control under $\alpha=0.01$.\label{tb::groups-fdr}}
\begin{center}
\begin{tabular}{lr}
\hline
Banks&Medical Products\\
Banks&Utility\\
Banks&Medical\\
Banks&Finance\\
Large Cap Pharma&Medical Products\\
Large Cap Pharma&Medical\\
Soap and Cleaning Products&Cosmetics \& Toiletries\\
Soap and Cleaning Products&Banks\\
Soap and Cleaning Products&Large Cap Pharma\\
Conglomerates&Aerospace\\
Conglomerates&Banks\\
Conglomerates&Transportation\\
Retail&Building Prds Retail\\
Wireless National&Banks\\
Building Products&Paper \& Related Products\\
Banks&Insurance\\
Building Products&Conglomerates\\
Conglomerates&Utility\\
Conglomerates&Machinery\\
Paper \& Related Products&Conglomerates\\
Building Products&Transportation\\
Large Cap Pharma&Banks\\
Business Services&Computer\\
Soap and Cleaning Products&Medical Products\\
Semi General&Computer\\
Conglomerates&Medical Products\\
Paper \& Related Products&Metal Products\\
\hline
\end{tabular}
\end{center}
\end{table}

\begin{table}
\caption{The industry groups and their associated stocks\label{tb::groups-stocks}}
\begin{center}

\begin{scriptsize}
\begin{tabular}{ll}
Metal Products & AA\\
Medical & AMGN\\
Steel & ATI\\
Cosmetics \& Toiletries & AVP\\
Medical Products & ABT BAX MDT\\
Utility & AEP AES ETR EXC SO\\
Insurance & AIG ALL CI HIG\\
Finance & AXP COF GS MS\\
Aerospace & BA GD RTN\\
Banks & BAC C JPM RF USB WFC\\
Large Cap Pharma & BMY JNJ MRK PFE\\
Beverages & CCU KO\\
Machinery & CAT\\
Soap and Cleaning Products & CL PG\\
Cable TV & CMCSA\\
Oil & COP CVX HAL SLB WMB XOM\\
Food & CPB\\
Computer & CSCO HPQ IBM MSFT ORCL PEP\\
Media Conglomerates & DIS\\
Auto & F\\
Transportation & FDX NSC UPS\\
Conglomerates & GE HON MMM UTX\\
Internet & GOOG\\
Building Prds Retail & HD\\
Semi General & INTC TXN\\
Paper \& Related Products & IP\\
Retail & MCD TGT WMT\\
Tobacco & MO\\
Industrial Robotics & ROK\\
Wireless National & S T VZ\\
Building Products & WY\\
Business Services & XRX	\\
\end{tabular}
\end{scriptsize}
	
\end{center}
\end{table}

\bibliographystyle{biometrika}
\bibliography{cond-depend}
\end{document}